\documentclass[aps,twocolumn,superscriptaddress,showkeys,thelongbibliography]{revtex4-2}
\usepackage{graphicx}

\usepackage{amssymb}
\usepackage{amsmath}
\usepackage{multirow}
\usepackage{placeins,physics}
\usepackage{xcolor}

\usepackage{etoolbox}
\usepackage{amsthm}
\newtheorem{definition}{Definition}
\newtheorem{theorem}{Theorem}
\usepackage{longtable}

\begin{document}

\title{Generating quantum channels from functions on discrete sets}

\author{A. C. Quillen}  
\email{alice.quillen@rochester.edu\\}
\email{\url{https://orcid.org/0000-0003-1280-2054} }
\affiliation{Department of Physics and Astronomy, University of Rochester, Rochester, NY 14627, USA}
\author{Nathan Skerrett}
\email{nskerret@u.rochester.edu}
\affiliation{Department of Physics and Astronomy, University of Rochester, Rochester, NY 14627, USA}

\begin{abstract}
Using the recent ability of quantum computers to initialize quantum states rapidly with high fidelity, we use a function operating on a discrete set to create a simple class of quantum channels.  Fixed points and periodic orbits, that are present in the function, generate fixed points and periodic orbits in the associated quantum channel.    Phenomenology such as periodic doubling is visible in a 6 qubit dephasing channel constructed from a truncated version of the logistic map.   Using disjoint subsets, discrete function-generated channels can be constructed that preserve coherence within subspaces.   Error correction procedures can be in this class as syndrome detection uses an initialized quantum register.  A possible application for function-generated channels is in hybrid classical/quantum algorithms.  We illustrate how these channels can aid in carrying out classical computations involving iteration of non-invertible functions on a quantum computer with the Euclidean algorithm for finding the greatest common divisor of two integers.
\end{abstract}


\maketitle

\tableofcontents

\section{Introduction}
\label{sec:intro}

Recent work has improved the fidelity and decreased the required duration
for quantum state initialization \citep{Yoshioka_2021,Johnson_2022}.
That makes it increasingly possible to use quantum state initialization as a component 
 of computations and algorithms that can be carried out on a quantum computer. 
Universality on a quantum computer encompasses non-unitary 
operations by leveraging interactions with a reservoir or environments allowing measurements \cite{Lloyd_2001,Verstraete_2009}.  
Algorithms that leverage the ability to carry out
 quick and accurate state initialization within an algorithm are predominantly non-unitary
 quantum operations and are a subset of possible quantum computations that
are facilitated via interaction with an external system.  
Yet, because of their simplicity, they are potentially useful as
 building blocks of more complex 
algorithms for the larger class of non-unitary quantum operations. 

We describe initialization of a quantum subsystem 
 with a quantum operation or channel from the space of density operators
of the quantum subsystem to itself; 
\begin{equation}
{\cal E}_{\rm init} (\rho) = \ket{\Psi_0}\!\bra{\Psi_0}, \label{eqn:init}
\end{equation}
 where 
$\rho$ is a density operator and $\ket{\Psi_0}$ 
 describes the initialized quantum state. 
We restrict our study so that quantum operations are unitary except initialization, 
so interaction with a reservoir or external environment effectively only occurs within initialization routines which could take place multiple times within an algorithm.  
During initialization, quantum subsystems are unentangled and superposition between states can be lost. 
Dynamics resembling that in classical systems, such as dissipation,  
 arises through loss of information that occurs when quantum states are reset. 

For a single qubit (a two state quantum system),  
all quantum operations (also called channels) can be described as affine transformations
of the polarization vector on the Bloch sphere (e.g., \cite{Nielsen_2010}).  However, the complexity of 
quantum channels in higher dimensional systems can exhibit 
more complex phenomena including attracting periodic orbits
or cycles, sometimes called rotating points \cite{Wolf_2012,Albert_2019}.

\subsection{Contributions}

We introduce a straightforward class of quantum channels that is generated by constructing Kraus operators from 
a function that operates on a finite set.  
We show how this class of channels can be implemented via unitary operations and state initialization.   We illustrate that these channels allow non-invertible functions to 
be iterated, cleanly delineate 
the asymptotic subspace from decaying subspaces, can exhibit cycles, can be completely
dephasing or can maintain coherence within a subspace.    As an example of
the potential complexity of these channels we illustrate phenomena of a dephasing
channel that is generated from the logistic map and exhibits much of the complexity of the logistic map. 


The class of channels we explore here would be useful as simple examples for 
illustrating phenomena that can be exhibited by quantum channels, including  
 cycles \citep{Wolf_2012,Albert_2019,Carbone_2020,Amato_2022}.
Hybrid classical/quantum algorithms  (e.g., \cite{Callison_2022}) might benefit from the ability to carry out classical algorithms on a quantum computer, 
without transferring information
back and forth between quantum and classical computers.  As an example of a 
classical algorithm that could be implemented on a quantum computer, we
show how to compute the greatest common divisor, a component of Shor's factoring
algorithm, with a dephasing function-generated channel. 

When a function-generated channel is not completely dephasing, coherence 
within a subspace can preserved. 
The amount of information lost is dependent upon the size of the subspace that is initialized
when the channel is implemented. 
We show that common quantum error correction protocols are found within our
class of quantum channels. 
Subspace preserving channels could be used as part of error recovery within
error correction procedures (e.g., \cite{Sivak_2023}).  

Lastly we list properties of all possible channels that fit within our class of function-generated channels for two and three state systems.   The examples are used to prove that the
eigenvalues of the matrix representation of a channel within this class are either zero or 
roots of unity. 


\subsection{Notation, definitions and conventions}

For notation and definitions we follow \cite{Wolf_2012}. 
We denote a finite dimensional complex vector space (a Hilbert space) as ${\cal H}$. 
The space of linear operators acting
on $\cal H$ is denoted by ${\cal B}({\cal H})$.
A quantum channel ${\cal E}$ is a linear,
completely positive, trace-preserving map from ${\cal B}({\cal H})$ to ${\cal B}({\cal H})$.
A density operator $\rho \in {\cal B}({\cal H})$ is a positive semidefinite operator
of trace 1. 
A set of Kraus operators $\{ K_i \} $ satisfies $\sum_i K_i^\dagger K_i = I$
where $I$ is the identity operator on $\cal H$, 
and gives a positive operator value measurement (POVM) or a quantum channel
via ${\cal E}(\rho) =   \sum_i K_i \rho K_i^\dagger$. 
The number of Kraus operators in a quantum channel, $n_K$, is called the Kraus rank.
The adjoint of a channel ${\cal E}^\ddagger $ is a linear 
map on ${\cal B}({\cal H})$ that is generated with the same set of Kraus operators; 
 ${\cal E}^\ddagger(X) = \sum_i K_i^\dagger  X K_i$.  
We use $\rho$ for an operator that has a trace of unity and $X$ otherwise. 

States in $\cal H$ for a  basis 
are specified by an integer $x \in {\mathbb Z}_N$, for example $\ket{x}$, 
where the set ${\mathbb Z}_N = \{0, 1, 2, ...., N-1 \}$, and  $N = {\rm dim} {\cal H}$.
For a pure state $\ket{\psi}\in {\cal H}$, the associated density operator is denoted $\ket{\psi}\!\bra{\psi}$. 
When working with bipartite systems we refer to quantum system $A$  
and quantum system $B$ and tensor product 
 space ${\cal H}_A \otimes {\cal H}_B$.
 States and density operators belonging to 
different quantum spaces are designated with subscripts; for example 
$\ket{\psi}_{\!A}$, $\rho_{\!A}$,  or $\ket{\psi}_{\!AB}$.

We define the dephasing quantum channel ${\cal E}_{dp}$ in an $N$ dimensional quantum space, 
$\dim { \cal H} = N$,  as that described with the set of $N$ Kraus operators, $\{ K_{dp,i} \}$ with 
$K_{dp,i} = \ket{i}\!\bra{i}$ with $i \in {\mathbb Z}_N$. 
This channel ${\cal E}_{dp}$ sends off-diagonal elements
to zero and diagonal matrix elements are preserved;  
${\cal E}_{dp}(\ket{i}\!\bra{j}) = \delta_{ij}$. 
A channel $\cal E$ is described as {\it dephasing}  if ${\cal E}_{dp} ({\cal  E}) = {\cal E}$. 
 
 A channel is described as {\it unital} if ${\cal E}(I) = I$
 and as {\it ergodic} if it has a unique fixed point $\rho_* = {\cal E}(\rho_*)$. 
 A channel is described as {\it mixing} if it is ergodic and all orbits
 decay to the fixed point; $\lim_{n\to \infty} {\cal E}^n(\rho)  = \rho_*$. 
 Here ${\cal E}^j(\rho)$ refers to the channel applied iteratively  $j$ times 
 ${\cal E}({\cal E}( ... ({\cal E}(\rho) ))$. 
 
Since a quantum channel is a linear map, it can be described with a matrix representation, 
called the Liouville representation, 
that operates on a vectorized version of an operator in  ${\cal B}({\cal H})$.  
We refer to the matrix representation of the channel $\cal E$ as ${\cal L}_{\cal E}$. 
With $N = {\rm dim} {\cal H}$, 
a linear operator in ${\cal B}({\cal H})$  is an $N \times N$ square matrix which can be written as an $N^2$-dimensional vector. The matrix ${\cal L}_{\cal E}$ is the $N^2\times N^2$ matrix that describes the operation of the channel $\cal E$ on the vectorized version of an operator. 
The matrix ${\cal L}_{\cal E}$ has elements 
${\cal L}_{{\cal E},\alpha \beta} = \tr(F_\alpha^\dagger {\cal E}  (F_\beta) )$ where
$F_\alpha, F_\beta$ are in the set of matrix units 
$ \{ \ket{i}\bra{j}: i,j \in 0, 1,2, ....., N \}$ (following \cite{Wolf_2008}).
The matrix ${\cal L}_{\cal E}$ is in general not  
Hermitian or invertible so it cannot necessarily be diagonalized, however a singular value decomposition can be used to find its eigenvalues and left and right eigenvectors.   Right eigenvectors $R_k$ satisfy ${\cal L}_{\cal E} R_k = \lambda_k R$ for $\lambda_k$  and ${\cal E}(R_k)  = \lambda_k R_k$ for
the associated eigenvalue $\lambda_k$ . 
Right eigenvectors of ${\cal L }_{\cal E}$ need not correspond to physical states;  for example, they need not have trace of 1. 

The asymptotic subspace of a channel is the set of operators spanned by
the right-eigenvectors of ${\cal L }_{\cal E}$ that have eigenvalue with modulus 1. 
Operators that lie outside of (are perpendicular to) this subspace decay during 
iterations of the channel \citep{Albert_2019}.     
A  `conserved quantity' (following \cite{Albert_2019}) is a left-eigenvector of ${\cal L}_{\cal E}$ with eigenvalue that has modulus 1. 
A conserved quantity  $J$  satisfies ${\cal E}^\ddagger (J) = e^{i\Delta} J$ for real $\Delta$ and  $\tr (J \rho) = e^{i\Delta} \tr  (J {\cal E}(\rho)) $ for any density operator $\rho$.     
If $\Delta$ labels an eigenvalue of modulus 1 and $\mu$ an index that counts degeneracies, then left and right  
eigenvectors of modulus 1 can be constructed so that they are bi-orthogonal;  
$\tr (J^{\Delta \mu} R_{\Delta' \mu'}) = \delta_{\Delta\Delta'}\delta_{\mu\mu'}$ 
\citep{Albert_2019}.

We describe the orbit of a function $f$ associated with initial condition $x_0$
as the set of points $\{ x_0, f(x_0), f^2(x_0), .... \} $. 
The orbit of a channel associated with an initial operator $X$ 
is the set $\{ X, {\cal E}(X), {\cal E}^2(X), {\cal E}^3(X), ..... \}$. 

Operators in  $ {\cal B}({\cal H})$ can be decomposed as a sum of basis elements that are orthogonal with respect to the Hilbert-Schmidt or Frobenius inner product 
$\langle A , B \rangle=\tr(A^\dagger B)$.  The associated norm is  
 called the trace distance $|| A|| =  \sqrt{\tr (A A^\dagger)}$.

\section{Generating iterates of a function on a quantum computer}
\label{sec:main}
 
Consider the quantum circuit illustrated in Figure \ref{fig:it}.
The full Hilbert space is ${\cal H}_A \otimes {\cal H}_B$ is a tensor product of two spaces.  Here $A$ is the first or top system or register (if it is composed of qubits) and 
$B$ is the second or bottom quantum system or register. 
The dimensions ${\rm dim} {\cal H}_A = {\rm dim } {\cal H}_B = N$.
The operator $U_f$ represents a unitary operation commonly called a `quantum oracle'
in the context of quantum black box algorithms where a function 
with hidden properties is queried \citep{Nielsen_2010}.
The oracle performs 
\begin{equation}
U_f : \ket{x}_{\!A}\ket{y}_{\!B} \to \ket{x}_{\!A}\ket{f(x)\! +\! y}_{\!B}  \ \  \ \  x, y \in {\mathbb Z}_N
\label{eqn:oracle}
\end{equation}
and $f()$ is a function from the set  ${\mathbb Z}_N$ 
to the same set.    Here $f(x) + y$ is modulo $N$. 
Equivalently,  
\begin{equation}
U_f = \sum_{x,y=0}^{N} \ket{x}_{\!A} \ket{y  +\! f(x)}_{\!B} \bra{x}_A \bra{y}_B.  \label{eqn:Uf}
\end{equation}

\begin{figure}[ht]
\includegraphics[width = 3 truein]{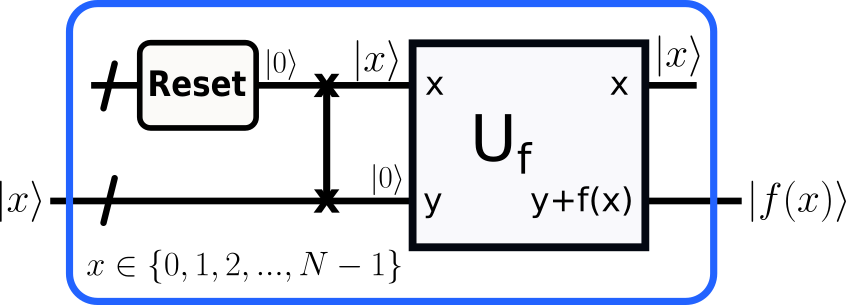}
\caption{This circuit creates a quantum channel operating on the bottom register.
The reset function initializes the top register. Both registers have $N$ states.
The large black box shows a unitary operation known as a quantum oracle (see equation \ref{eqn:oracle}) that depends on a function $f$.  
In the center of the circuit is a swap operation (equation \ref{eqn:swap}). 
The initial state on the bottom register is a pure state $\ket{x}$ 
in the basis labelled  with $x \in \{0, 1, ...,  N-1\}$.  
When iterated, 
the operation within the blue outer box generates iterates of the oracle function $f()$.
 \label{fig:it}
}
\end{figure}

The box labelled {\it Reset} in Figure  \ref{fig:it} initializes 
the top register so that it becomes $\ket{0}^{\otimes n}$.
In the middle of Figure \ref{fig:it} there is a swap operation which replaces 
the top register with the bottom register and vice versa; 
\begin{equation}
{\rm SWAP} = \sum_{x,y=0}^{N-1} \ket{x}_{\!A}\ket{y}_{\!B} \bra{y}_A\bra{x}_B .\label{eqn:swap}
\end{equation}
The swap operation allows us to have $\ket{x}$ input and oracle output $\ket{f(x)}$ on the same B register.  
For the moment we assume that the input $\ket{x}_{\!B}$ is a pure state in the 
basis labelled by integer index. 
This ensures that the output state $\ket{x}_{\!A}\otimes \ket{f(x)}_{\!B}$ is not entangled. 

Consider the operation within the blue box.  If the input in the B register 
is a state $\ket{x}_{\!B}$
then the output is the state $\ket{f(x)}_{\!B}$.  If the operation is repeated 
then the output is $\ket{f(f(x))}_{\!B}$.  By calling the operation multiple times
it is possible to create quantum states from iterates of the oracle function.  
As the function need not be invertible, this operation is not necessarily unitary. The reset is used to remove information from the system that would allow the operation to be inverted. 

We consider functions $g$ that map real numbers from the unit interval $[0,1)$ onto to the unit interval. 
If the system  is comprised of $n$ qubits then $N=2^n$. 
The oracle function $f: {\mathbb Z}_N \to {\mathbb Z}_N$.
We can create the oracle function $f$ from $g$ by truncation. 
Following \cite{Miyazaki_2011,Miyazaki_2012}  
\begin{equation}
 f(x) =  {\rm floor} \left[ g\left( \frac{x}{2^n} \right) 2^n \right] \ \ \  {\rm for} \ x \in \{ 0, 1, ..., 2^n\!-\!1\}.  \label{eqn:fg}
\end{equation}
The floor function takes as input a real number $y$, and gives as output the greatest integer less than or equal to $y$.
Consider the binary fraction given
by the digits of non-negative integer $x \in \{0, 1, 2, ..., 2^n\! -\!1\}$. The state
$\ket{x}$ can also be labelled with its binary digits,
$\ket{x_1 x_2 x_3 x_4 ... x_n}$,   
with $x_j \in \{0,1\}$ and $x = \sum_{j=1}^n x_j 2^{n-j}$.  A number on the 
 the unit interval $y \in [0,1)$ can be created from $x$ via $y=\sum_{j=1}^n x_j 2^{-j}  = x/2^{n}$.
Thus equation \ref{eqn:fg} equivalently gives a map
 from binary fractions to binary fractions. 
 Truncation is essentially equivalent to zeroing binary digits past a particular index. 
Equation \ref{eqn:fg} truncates the function $g$ on the unit interval to the accuracy 
$1/2^n$. 

\subsection{Iterates of the truncated logistic map}
\label{sec:trunc}

\begin{figure}[ht]
\includegraphics[width = 3 truein]{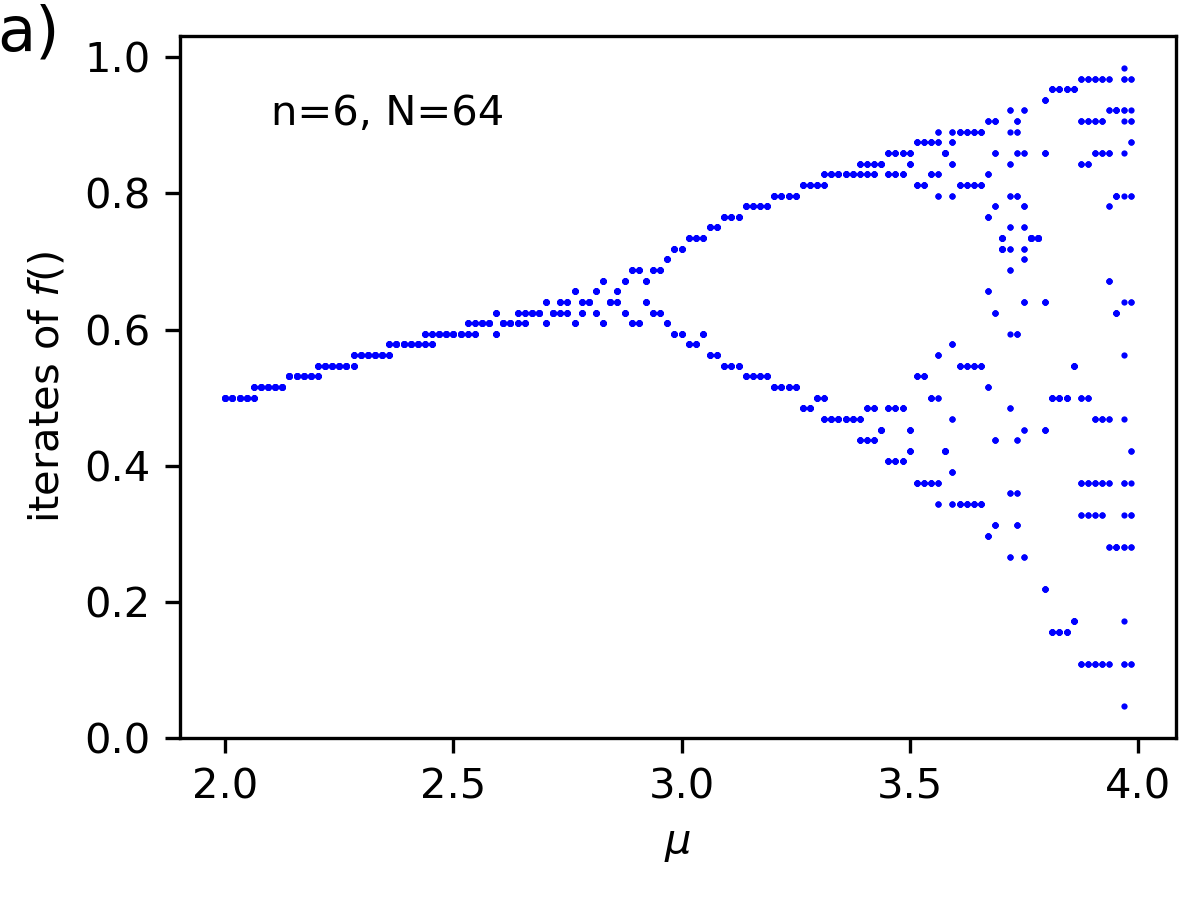}
\includegraphics[width = 3 truein]{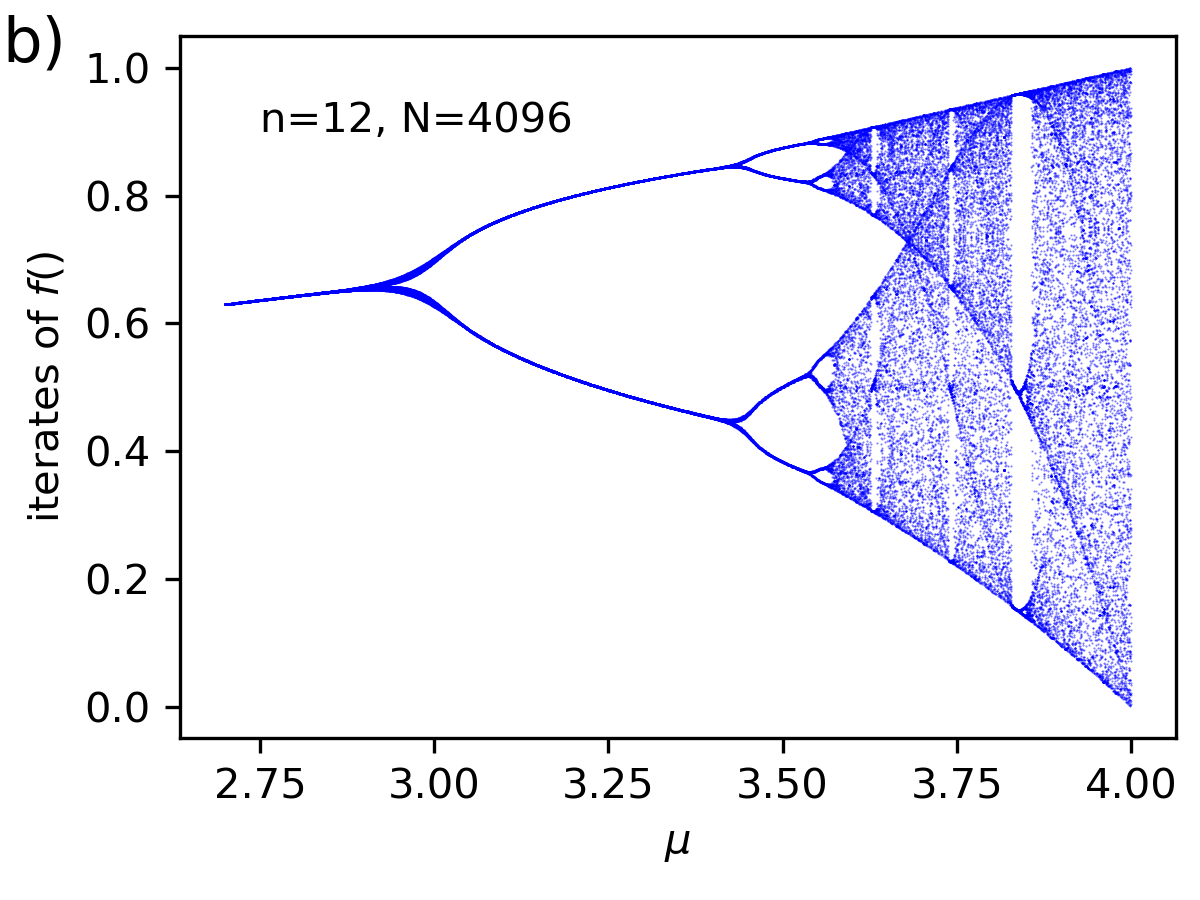}
\caption{The result of iterating the circuit shown in Figure \ref{fig:it}
with a) $n=6$ qubits (top) and b) $n=12$ qubits (bottom).
The oracle function is generated by truncating the logistic map and is 
given by equation \ref{eqn:fg}. For both panels, the initial condition
$x_0 = 2^{n-1}$. 
In a) we iterate 
20 times and then plot 20 iterates.  The x-axis is restricted to $\mu >2$. 
In b) we iterate 30 times and then plot 20 iterates.  The x-axis range is restricted to 
$\mu >2.7$ so as to better show the region with substructure.   
The iterates of the truncated logistic map display many properties of the 
logistic map. 
\label{fig:log}}
\end{figure}

To illustrate potential complexity of a channel created via the circuit
in Figure \ref{fig:it}, we create
quantum channel based on the logistic map.   We chose the logistic map
because its iterates exhibit a range of phenomena \cite{May_1976} and 
because it is based on a quadratic function and can be computed
efficiently with addition and multiplication routines for integers on qubit systems 
(see section 6.4.5 by \cite{Rieffel_2011} and \cite{Seidel_2022}). 
The logistic map is the function 
\begin{equation}
g(x) = \mu x(1 -x) \qquad {\rm for} \ \  x \in [0,1]. \label{eqn:logistic}
\end{equation}
The logistic map is sensitive to the parameter $\mu \in [0, 4]$
with range restricted so that iterates of the function remain in the unit interval. 

To create the associated oracle function $f$ via equation 
\ref{eqn:fg}, we restrict $\mu$ to binary fractions with the
same resolution as available on the quantum register used to store $x$.  
In other words we chose $\mu$ so that $2^n\mu$ is an integer.  

In Figure \ref{fig:log} we show the result of iterating the oracle
with oracle function generated by truncating the logistic map. 
In Figure \ref{fig:log}a we show iterates $20$ through 40 for $n=6$ qubits (and $N=2^n=64$ states).
In Figure \ref{fig:log}b we show iterates $30$ through 50 with $n=12$ qubits (corresponding to $2^n = 4096$ states). For both panels, the initial condition
$x_0 = 2^{n-1}$. 
The x axes show $\mu$, controlling the logistic map and its truncated version. 
The y axes show iterate values $f(x)/2^n$ so they lie within the unit interval. 
For both figures the initial condition is given by $x_0 = 2^{n-1}$. 

The orbits of the logistic map are insensitive to the initial condition
(excepting a small set  which includes the initial condition $x=0$).
For $\mu < 1$  all orbits approach $x=0$. For $1 < \mu < 3$ there is 
a single attracting fixed point.  At $\mu = 3$ there is a bifurcation
giving an attracting period 2 orbit for $3.0 < \mu \lesssim 3.45$. 
Subsequent bifurcations show
the phenomenon known as {\it period doubling}. 
Despite a low number of qubits, the structure of the logistic map is present
in the iterates shown in Figure \ref{fig:it}, of the truncated function of equation \ref{eqn:fg}.
 
In Figure \ref{fig:log}a for $\mu$ near but below 3, the distributions of iterates 
is wider than at lower values of $\mu$. For each value of $\mu$, the iterates are close
together but cycle between nearby values.   When an orbit is periodic,
we refer to it as a cycle.  
Maps can have computationally determined periods that are dependent on the precision of
the arithmetic \citep{Chia_1991,Miyazaki_2012}.  We see this phenomenon here
as iterates in the truncated map can cycle between nearby states
for the same value of $\mu$ giving iterates in the original map that approach a fixed point. 
For the logistic map, if $\mu$ gives a periodic attractor, orbits
from almost all initial conditions 
converge on to it.  However, for the truncated map,  at a given $\mu$, 
almost all initial conditions will give orbits that enter periodic cycles with the same period,
however the points in these different cycles may not all be identical. 

The sensitivity to initial condition of the period $p$ of a cycle that an orbit enters 
is illustrated in more detail in Figure \ref{fig:combos}. 
We compute iterates of the oracle function (the orbits) 
for different values of initial condition $x_0$ and control variable $\mu$. 
All orbits eventually enter a periodic cycle. The index of
the iteration where the cycle is entered we call $k_c$. 
This number is referred to as the {\it link length} by \cite{Miyazaki_2012} and defined as 
\begin{equation}
k_c(x_0) = {\rm min} \{k | f^k(x_0) = f^{k+p} (x_0) {\rm\  for \ } k \ge 0\}.  \label{eqn:k_c}
\end{equation}
where $p$ is the period of the cycle that is entered. 
The median period
of the cycle for different initial conditions is shown as a function of $\mu$
in the top panels in Figure \ref{fig:combos}.   The period of the cycle that
an orbit enters is shown  
shown as an image and as a function of $\mu$ (on the x-axis) and initial condition
$x_0$ on the y-axis in the second panels. 
In the third panels we plot the iterates after the cycle has been entered and 
with color set by the cycle period.  
The link length, or iteration number when 
the cycle is entered, $k_c$ is shown as an image in the bottom panels. 

Due to truncation, 
the attracting orbits of $f$  either remain in a single state or are periodic, 
even for values of $\mu$ that would give chaotic orbits
in the logistic map.  For values of $\mu$ that give chaotic orbits in the logistic map,
the truncated map $f$ could have a long period, with period
 sensitive to the truncation level, which is set by the number of qubits in the quantum  registers. 
Figure \ref{fig:combos}b shows that for $n=7$ all orbits have entered a cycle
within only 30 iterations.  The maximum cycle period is 23.  This suggests that 
both the maximum period and the maximum
number of iterations required to enter a cycle 
scale logarithmically with the total number of states; $N=2^n$. 
While we find that period 3 orbits do not exist for $n\le 5$, they do exist for $n>5$. 
As the logistic map can exhibit  cycles of any period,  
it is likely that with sufficient numbers of bits, it is possible
to find a value of control parameter $\mu$ and an initial condition
that gives an orbit of the truncated logistic function 
that converges onto a cycle with any desired period. 

In the logistic map, past $\mu=3$, period 1 orbits (fixed points) are unstable. 
However the bottom panels of Figure \ref{fig:combos} show
that for $3 < \mu \lesssim 3.5$, orbits for most initial conditions enter
period 2 cycles, however a few initial conditions give orbits that enter period 1 cycles (the same thing as fixed points). 
The truncated map can have stable and attracting fixed points that are not present in the logistic map. 
 
In this section we have numerically illustrated that the orbits of the truncated logistic map 
have remarkable complexity, even when the level of truncation is fairly large. 
The truncated logistic map could be implemented on a quantum computer
using a modular multiplication circuit (e.g., \cite{Seidel_2022}), 
that is constructed of gates that operate on only a few qubits, 
such as the controlled NOT, NOT and Toffoli gates.
So if the map $f$ is generated from simple arithmetic operations
(as is the logistic map), 
then it is straight-forward to implement on a quantum computer. 
The complexity of maps such as the truncated logistic map could add to 
the repertoire of circuits that can be performed on a quantum computer. 
 
 \begin{figure*}[ht]
\includegraphics[width = 3 truein]{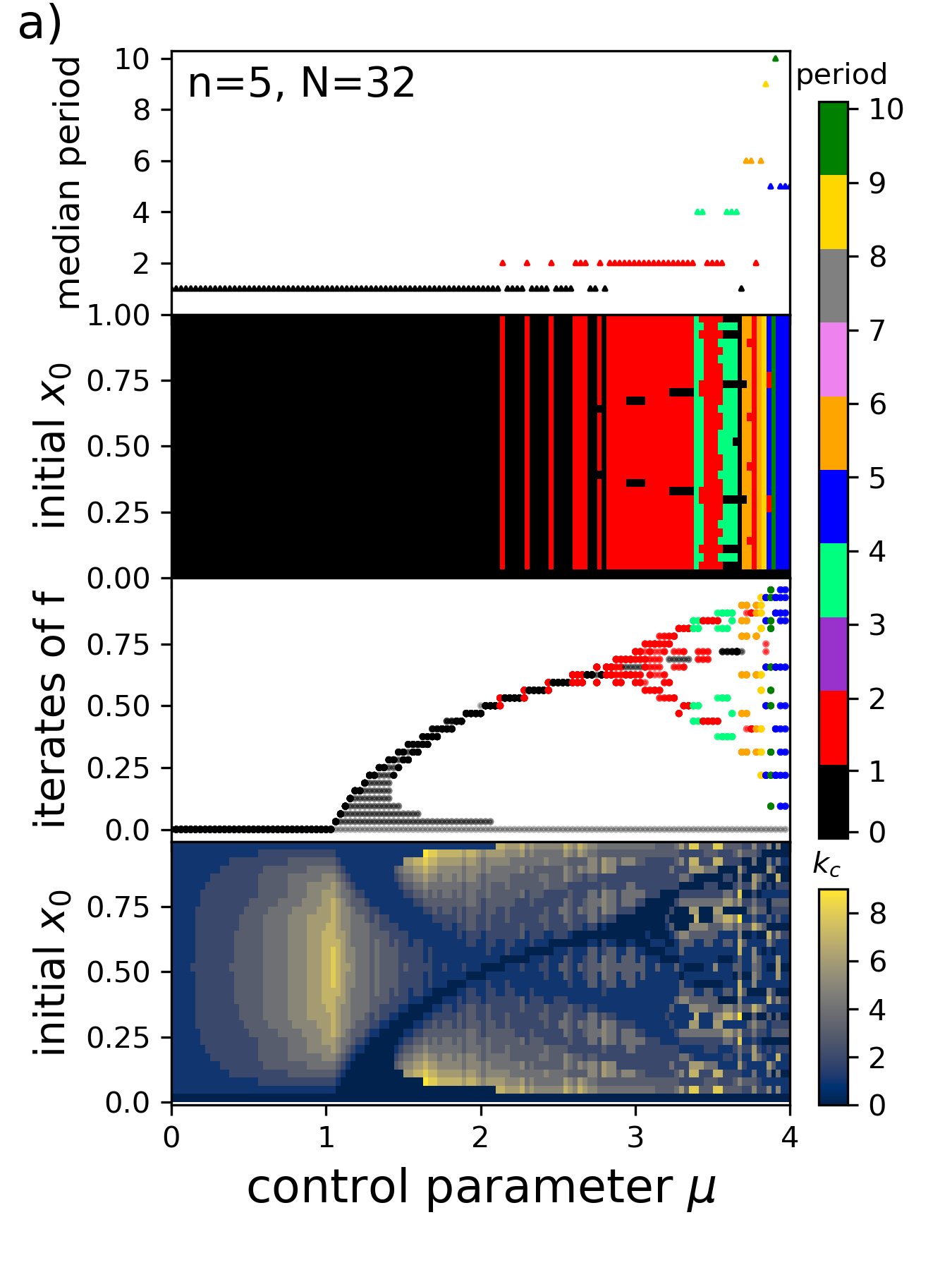}
\includegraphics[width = 3 truein]{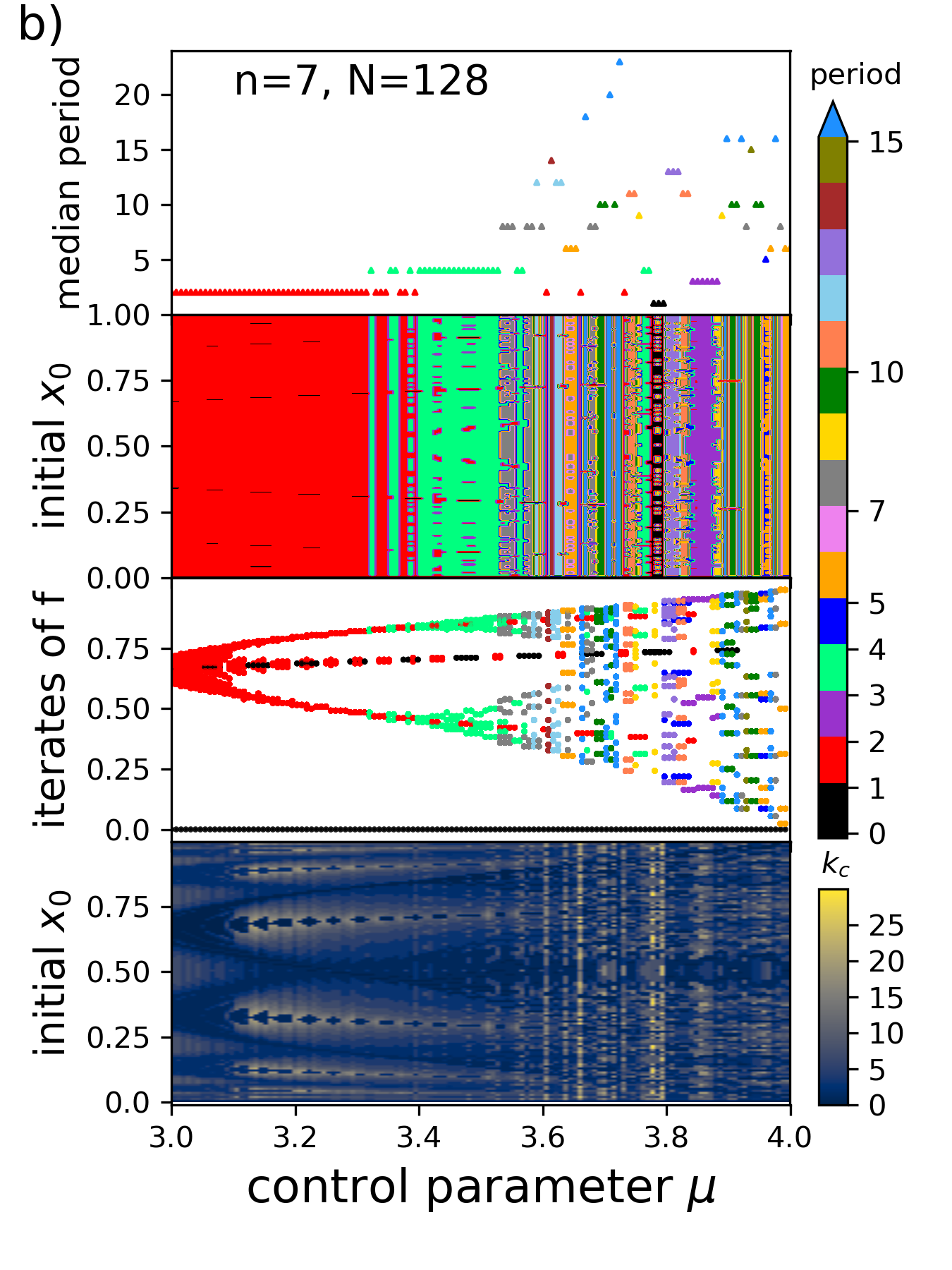}
\caption{The result of iterating the circuit shown in Figure \ref{fig:it}
with a) $n=5$ qubits and using a truncated version of the logistic map as oracle function 
(equation \ref{eqn:fg}).
We compute orbits for different integer initial conditions $x_0$ and values of $\mu$ the control parameter for the oracle function.   
In the top panel we show a median period
as a function of $\mu$, on the x axis. The median is computed from the set
of attracting cycles averaged over the possible initial values $x_0$.
In the second panel from top, we show the period of the attracting
cycle as a function of  initial condition $x_0$ (on the y-axis) and  
control parameter $\mu$ on the x-axis.
In the third panel, the iterates are plotted as a function
of $\mu$, with color dependent on the period of the attracting cycle that the orbit
enters.    
We only plot iterates after the attracting cycle has been entered. 
In the bottom panel, the image shows the link length $k_c$, or iteration number where the cycle was entered, as a function of initial condition
and control parameter. 
b) Similar to a) except we show the case with $n=7$ qubits and we restrict
the $x$ axis so $\mu\ge 3$. 
\label{fig:combos}}
\end{figure*}

\subsection{The reset quantum channel}
\label{sec:reset}

In Figure \ref{fig:it} we assumed that the input to the bottom register
was a pure state $\ket{x}$ in the  basis with $x \in {\mathbb Z}_N$.
What happens if the input bottom register is a superposition or a mixed state or if
the two registers are entangled?

The reset operation appears to operate on the A register. 
However, we need to specify how the initialization affects the full system ${\cal H}_A \otimes {\cal H}_B$.
The density operator describing both registers is $\rho_{AB}$.
The density operator for the lower register is  
\begin{equation} 
\rho_B = \tr_A \rho_{AB}. \label{eqn:rhoB}
\end{equation} 

The reset or initialization function should act like a generalized or positive operator value measurement (POVM); (e.g., \cite{Kranz_2019,Basilewitsch_2021,Johnson_2022,Volya_2023}). 
For example, suppose we measure a single qubit via a measurement
associated with the Pauli $Z$ operator.
If the measurement gives state $\ket{0}$, we do nothing, corresponding
the identity operation, but if we measure the system
to be in the state $\ket{1}$, we then apply the Pauli X operator, putting the
system into the $\ket{0}$ state.  This procedure resets the qubit and is 
 described by a POVM measurement with two Kraus operators 
$K_0 = \ket{0}\!\bra{0}$, $ K_1 = \ket{0}\!\bra{1}$.
In a larger bipartite quantum system, resetting one subsystem 
is described by a quantum channel that is given by the set of $N$ Kraus operators; 
 \begin{equation}
 K_j  =\ket{0}_{\!A}\!\bra{j}_{\!A} \otimes I_B \qquad {\rm with\ }  j \in {\mathbb Z}_N. \label{eqn:K_j}
 \end{equation}
The set $\{ K_j \}$ satisfies $\sum_{j=0}^{N} K_j^\dagger K_j = I_{AB}$,
following the definition for a POVM, ensuring that the operation preserves the trace. 
The set $\{ K_j \}$ of Kraus operators give an operator sum (or Choi) representation for the reset operation.   In other words, 
the quantum channel ${\cal E}_{\rm reset}$ describing the reset is the 
 completely positive and trace preserving map 
 $\rho_{AB} \to \rho_{AB}'$:
 \begin{align}
 \rho'_{AB} & = {\cal E}_{\rm reset}(\rho_{AB})  = \sum_{j=0}^{N} K_j \rho_{AB} K_j^\dagger.  
  \end{align}
Using equation \ref{eqn:K_j} for the Kraus operators we find that 
  \begin{align}
 {\cal E}_{\rm reset}(\rho_{AB}) &= \ket{0}_{\!A}\!\bra{0}_{A} \otimes \rho_B  \label{eqn:reset}
 \end{align}
 with $\rho_B $ given by equation \ref{eqn:rhoB}.

It is a general principle that 
a POVM measurement in one subspace does not
affect the traced density matrix of another subspace (e.g, \cite{Nielsen_2010}). 
The reset channel described in 
equation \ref{eqn:reset} is consistent with this but also shows
that for an initial state that is entangled, after the reset is applied,
the resulting state is no longer entangled. 
Thus the reset channel is an example of an {\it entanglement breaking}
channel (e.g., section II.F in the review \cite{Caruso_2014}). 
If we only consider how the reset operation affects the A register 
then the reset channel maps $\rho_A \to  \ket{0}_{\!A} \bra{0}_A$
and is in the form of equation \ref{eqn:init}.

\subsection{Properties of the oracle channel}
\label{sec:oracle_chan}

In Figure \ref{fig:itb} we show a modification of 
the circuit of Figure \ref{fig:it}.  The reset is applied after
the oracle and drawn so that it shows more clearly that 
the reset operation ${\cal E}_{\rm reset}$ affects both registers
when they are entangled.    The set of operations in Figure \ref{fig:it} is 
a quantum channel ${ \cal E}_f$ on the second register, taking
$\rho_B \to \rho_B'$ via $\rho_B' ={ \cal E}_f( \rho_B)$.
We can call this an oracle channel as it is based on a quantum 
oracle operation. 

\begin{figure}[ht]
\includegraphics[width = 3 truein]{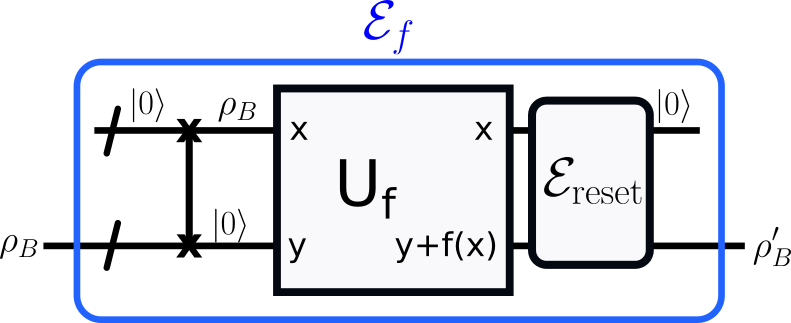}
\caption{This circuit creates a quantum channel ${\cal E}_f$ operating on the bottom register.
This circuit is similar to that shown in Figure \ref{fig:it} except we do not restrict the input
of the bottom register to be a basis state and the reset is done after the oracle instead of before. The input of the bottom register is described with density operator $\rho_B$.
The output in the bottom register is $\rho_B' = {\cal E}_f(\rho_B)$.   The reset operator, ${\cal E}_{\rm reset}$, 
initializes the top register and breaks the entanglement between the registers. 
 \label{fig:itb}
}
\end{figure}

Suppose we input a superposition state to the circuit shown in Figure \ref{fig:itb};
\begin{align}
\ket{\psi_0}_{\!AB} &= \ket{0}_{\!A} \otimes (a \ket{x}_{\!B}+ b \ket{y}_{\!B}) 
\end{align}
with $x\ne y$. The density matrix of the
second register is initially $\rho_B = \tr_A (\ket{\psi_0}_{AB}\!\bra{\psi_0}_{AB})$. 
After applying the swap and oracle, and prior to the reset, 
the density matrix becomes, 
\begin{align}
\rho_{AB,1} & = U_f  \ket{\psi_0}_{\!AB}  \bra{\psi_0}_{AB} U_f^\dagger  \nonumber \\
& = (a \ket{x}_{\!A} \ket{f(x)}_{\!B} + b  \ket{y}_{\!A} \ket{f(y)}_{\!B}) \times \nonumber \\
& \qquad ( a^* \bra{x}_A \bra{f(x)}_B + b^*  \bra{y}_A \bra{f(y)}_B).
\end{align}
We apply the reset (equation \ref{eqn:reset})
\begin{align}
 {\cal E}_{\rm reset}(\rho_{AB,1}) & = \ket{0}_{\!A}\!\bra{0}_A \otimes \rho_B',
\end{align}
where the output density operator 
\begin{align}
\rho_B' 
& = { \cal E}_f (\rho_B) 
 = \tr_B \rho_{AB,1}   \nonumber\\
& = aa^* \ket{f(x)}_{\!B} \!\bra{f(x)}_B  + bb^*\ket{f(y)}_{\!B}\! \bra{f(y)}_B . \label{eqn:oo}
\end{align}
If $f(x) =f(y)$ then the result $\rho_B' =\ket{f(x)}_{\!B}\!\bra{f(x)}_B$
due to normalization of $\ket{\psi_0}_{\!AB}$. 
Otherwise, with $f(x) \ne f(y)$, 
the output density operator, shown in equation \ref{eqn:oo}, 
is a mixed state composed of two pure states. 
This implies that ${\cal E}_f( \ket{i}_{\!B} \!\bra{j}_B) = 0$ for $i\ne j$. 
The channel is dephasing as it zeros off-diagonal elements of an operator. 


Using the Kraus operators for the reset channel 
(listed in equation \ref{eqn:K_j}), the oracle function (equation 
\ref{eqn:Uf}),
the swap operation (equation \ref{eqn:swap}), and setting the initial top register  
to $\ket{0}$, we can construct a complete set of Kraus operators for ${\cal E}_f$,  the 
channel for the whole operator in the large blue box in Figure \ref{fig:itb};
\begin{align}
{\tilde K}_j &= \ket{f(j)}_{\!B} \!\bra{j}_B  \qquad 
{\rm for \ } j \in {\mathbb Z}_N . 
\end{align}
The channel ${\cal E}_f $ is described by the map  $\rho_B \to \rho_B'$ with 
\begin{align}
\rho'_B = {\cal E}_f(\rho_B) =\sum_{j=0}^{N-1} {\tilde K}_j \rho_B {\tilde K}_j^\dagger.
\end{align}
The channel ${\cal E}_f $ is completely dephasing as it zeros all off-diagonal elements in
the density matrix. 
Application of the channel $k$ times to a diagonal mixed state $\rho = \sum_i p_i \ket{i}_B\bra{i}_B$ with $p_i$ a set of probabilities gives 
\begin{align}
{\cal E}_f^k \left(\sum_i p_i \ket{i}_{\!B}\!\bra{i}_B\right) 
   = \sum_i p_i \ket{f^k(i)}_{\!B} \! \bra{f^k(i)}_B. 
\end{align}
The channel iterates the function $f$, but essentially behaves
classically as it preserves the set of probabilities from the diagonal elements of the density operator. 


The oracle channel shown in Figure \ref{fig:itb} does not output superpositions 
of states in the B register,  
 however we can make other types of superpositions.
We can consider a function $f^k(\mu,x)$ where $\mu$ is the control
parameter and $k$ is the iteration number. 
For example we could create superposition states like $\sum_\mu \ket{\mu}\ket{f(\mu,x)}$
or $\sum_k \ket{k} \ket{f^k(\mu,x)}$.  As the reset operation only takes
the trace over the $\ket{x}$ degrees of freedom,  superpositions of states labelled by the control parameter $\mu$ or iteration number would remain intact. 


\subsection{Euclidean algorithm for finding the greatest common divisor}
\label{sec:Euclid}

The Euclidean algorithm for finding the greatest common divisor is a component
of the Shor factoring algorithm \citep{Shor_1999}.  If there is a penalty to transferring information
between classical and quantum computing architectures (error associated with measurements; e.g., \cite{Chen_2023}), 
there could be an advantage to performing this
operation, even though it is a classical one \citep{Callison_2022}, on a quantum computer.

The Euclidean algorithm for finding the greatest common divisor can be done
iteratively with a function that takes two integers  
\begin{align}
g_{gcd}(a,b) = & \begin{cases}
(b, a {\rm\ mod}\ b ) & {\rm\  if \ } b>1  \\
(a,0) & {\rm\  if \ }  b=0  \\
\end{cases}. \label{eqn:gcd}
\end{align}
If an orbit is initialized with $(a,b)$,  and $a>b$ then if $a,b$ are relatively prime,  iterates of the function eventually reach a fixed point $(1,0)$.
Otherwise, the end state is $(x,0)$ with $x$ equal to the greatest common divisor of $a,b$. 
The function need only be iterated $O(\ln(b))$ times to reach an end-state.
If an oracle operation is constructed to perform this function, then 
it would be possible to compute
the classical Euclidean algorithm on the quantum computer, using  
 the circuit shown in Figure \ref{fig:itb}.

For positive integers $a,b$, 
given $a \ {\rm mod} \ b $ and $a$ it is not possible to uniquely find $b$.   
Given $a \ {\rm mod} \ b $ and $b$ it is not possible to uniquely find $a$. 
Consequently functions $(a, b) \to (a \ {\rm mod} \ b, b)$ and
$(a, b) \to (a, a \ {\rm mod} \ b)$ are not invertible. 
However we can create a unitary transformation that
carries out a modulo function with divisor as argument with 3 quantum registers 
 \begin{equation}
U_{Mo} : \ \ \ket{a}\ket{b}\ket{c} \to  
\begin{cases} \ket{a}\ket{b} \ket{a \ {\rm mod}\ b + c} \ {\rm if \ } b \ne 0 \\
\ket{a}\ket{b} \ket{a  + c} \ {\rm if \ } b = 0 
\end{cases}.
\label{eqn:UMo}
\end{equation}
Here $(a \ {\rm mod}\ b + c)$ is assumed to be modulo $N$ with $N$
the dimension of each of the three quantum registers.  
This unitary operation is a quantum oracle, but the function has two arguments 
instead of one. 

Using $U_{Mo}$, the greatest common divisor of two integers can be 
computed iteratively using the channel ${\cal E}_{gcd}$ shown in Figure \ref{fig:gcd}.
The unitary operation $U_{Mo}$ is combined with two swap operations
and a reset of the bottom register.
 The channel computes 
 \begin{align}
 {\cal E}_{gcd}(&\ket{a}_{\!A}\! \bra{a}_{\!A}\! \otimes\!  \ket{b}_{\!B} \!\bra{b}_B) 
=    \\
&\begin{cases}
  \ket{b}_{\!A}\!\bra{b}_A\!  \otimes\!  \ket{a\ {\rm mod}\ b}_{\!B}\! \bra{a\ {\rm mod\ } b}_B  & 
  {\rm  if \ } b \ne 0  \\
 \ket{a}_{\!A} \!\bra{a}_A \!\otimes \! \ket{0}_{\!B} \!\bra{0}_B & {\rm if \ }  b = 0 \\
 \end{cases}.\nonumber
 \end{align}
 Again the channel is completely dephasing, so off-diagonal elements are zeroed. 

When initialized with $\ket{a}\ket{b}$, for $a >b>0$ and then ${\cal E}_{gcd}$ is called iteratively, the system 
reaches an end-state with top register equal to the greatest common
 divisor  of $a,b$; 
 \begin{align}
 \lim_{k \to \infty} {\cal E}_{gcd}^k(\ket{a}_{\!A}\!\bra{a}_A\! \otimes\! \ket{b}_{\!B}\!\bra{b}_B)
  =&  \ket{ {\rm gcd}(a,b) }_{\!A}\! \bra{ {\rm gcd}(a,b) }_A \nonumber \\
  & \otimes\! \ket{0}_{\!B}\!\bra{0}_B.
 \end{align}
Because the Euclidean algorithm reaches an end-state in at most $O(\ln b)$ iterations,
the channel reaches its end-state with iteration number $k =  O(\ln b)$.
The complexity of the algorithm rests upon how many gates are needed to create 
 the unitary operator $U_{Mo}$ of equation \ref{eqn:UMo}.  
Computation of the greatest common divisor could be done with serial subtraction
and conventional arithmetic quantum circuits.  
Iteration of the channel does not significantly
add to the computational complexity of computing the quantum oracle function itself.
Implementation of this channel would not significantly add to the complexity of Shor's algorithm itself.   
Future implementations of Shor's algorithm on large qubit systems could include an iterative computation of the greatest common divisor on the quantum computer.
It is amusing to reflect that we have illustrated a possible implementation of a 2300 year old classical algorithm on a quantum computer.  

\begin{figure}[ht]\centering
\includegraphics[width=3truein]{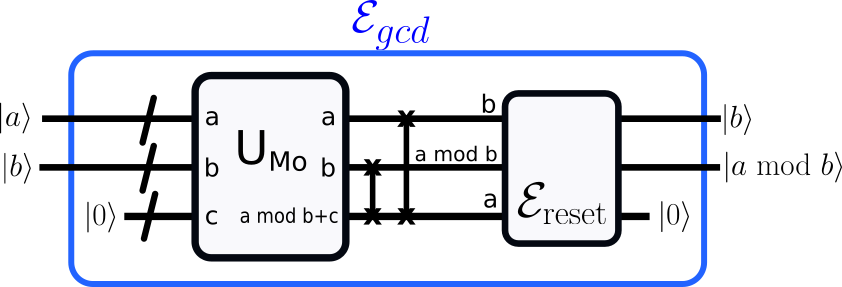}
\caption{A circuit to compute the greatest common divisor 
of two integers via a quantum channel.
This quantum operation performs the function shown in 
equation \ref{eqn:gcd}.  The unitary operation $U_{Mo}$ is that of equation \ref{eqn:UMo}.
The reset operation resets the bottom register. 
When initialized with  $\ket{a}\ket{b}$, for $a >b$,  
and ${\cal E}_{gcd} $ is called iteratively, 
the system reaches a fixed state with  top register equal to $\ket{{\rm gcd}(a,b)}$. 
\label{fig:gcd}
}
\end{figure}

\section{Quantum channels constructed from functions operating on discrete sets}

The channels introduced in section \ref{sec:main} are completely dephasing
in the sense that they zero all off-diagonal elements of the density operator. 
We can create a quantum channel based on a function $f: {\mathbb Z}_N \to {\mathbb Z}_N$ that is somewhat more flexible by 
using disjoint subsets of ${\mathbb Z}_N$.   This allows us to create
a channel that has fewer than $N$ Kraus operators and  
that can preserve superposition or coherence within
certain subspaces. 
Implementation of the channel involves initializing a smaller quantum space  
and so reduces the information lost.  

\subsection{A definition for a discrete function-generated channel that is based
on disjoint sets}

From a discrete function and a collection of disjoint sets,  we define a particular quantum channel. 

\begin{definition} \label{def:chan}
Consider a function  $f: {\mathbb Z}_N \to {\mathbb Z}_N$ and a collection
of disjoint sets 
$S_D = \{ S_j \}$ with $j \in 0, 1 ..., n_K-1$.  Each set contains integers in ${\mathbb Z}_N$ and is labelled with an integer. We require the sets to be disjoint $S_j\cap S_k = \emptyset$ for all $j \ne k$,   have  union 
 $\cup_j S_j = {\mathbb Z}_N$ and  also satisfy   
 $|f(S_j)| = |S_j|\ \forall S_j$.  Here    
$f(S_j)$ refers to the set $f(S_j) = \{ f(k) :  k\in S_j \}$,  
and $n_K = |S_D|$ is the number of disjoint sets. 

We construct a quantum channel ${\cal E}_{S_D, f}$ that operates
on a finite Hilbert space of dimension $N$.  The channel is described
with the set of $n_K$ Kraus operators $\{ K_j \}$ 
\begin{align}
K_j  = \sum_{i \in S_j} \ket{f(i)}\!\bra{i} .  \label{eqn:ES}
\end{align}

\end{definition}

 
To ensure that the resulting channel is a trace-preserving, completely positive map,  we check that the preceding definition gives a complete
set of Kraus operators,
\begin{align*}
K_j^\dagger K_j  &= \sum_{i,k \in S_j} \ket{i}\!\bra{f(i)} \ket{f(k)}\!\bra{k} \\
& = \sum_{k \in S_j}  \ket{k}\!\bra{k}  + \sum_{i, k\in S_j; i \ne k, f(i) = f(k) } \ket{i}\!\bra{k}.
\end{align*}
The term on the right is zero via the pigeon hole principle because $|f(S_j)| = |S_j|$ (which is part of our definition).
Thus 
\begin{align*}
\sum_j K_j^\dagger K_j = \sum_j \sum_{k\in S_j} \ket{k}\!\bra{k} = I
\end{align*}
because the sets are disjoint and have union equal to ${\mathbb Z}_N$, and
$\ket{k}$ with $k \in {\mathbb Z}_N$ are an orthonormal basis for our quantum space. 

\subsection{Realizing the channel on a quantum computer} 
\label{sec:realize}

We describe how we can modify the circuit
discussed previously (shown in Figure \ref{fig:itb})  so that we can construct 
 a quantum circuit that performs the channel with
Kraus operators defined in equation \ref{eqn:ES}, based on 
a function $f$ and a collection of subsets $S_D$ defined as in definition \ref{def:chan}.
Instead of using a register that has number of states $N$, 
we use an additional register that has $n_K$ states equal to the number of disjoint sets which is equal to the Kraus rank of the channel. 
We denote the additional register as the S-register. The register for the channel
itself is again the B-register.  The circuit is outlined in Figure \ref{fig:gen}.

Because the sets $S_j$ are disjoint, they are associated with a collection
of orthogonal projection operators  $\{ P_i \}$, one for each disjoint set
\begin{align}
P_j = \sum_{k\in S_j} \ket{k}\!\bra{k}. \label{eqn:Pj}
\end{align}
 Because the union of the subsets
span the quantum vector space and $\sum_j P_j = I$, the set $\{ P_j \} $ gives a  projective measurement. 
For each subset $S_j$, we 
 could assign a single qubit in the S register and use
controlled operations to  mark specific states in the B register.  Then controlled gates can be used to carry out operations on the B register that depend upon the state in the S register.  It is more efficient in terms of space (the needed number of states in the S register)
to assign each subset to a single state in the S register rather than a single qubit in this register. 

We construct a unitary operator to mark states in the B register 
according to which subset in $S_D$ they belong to
\begin{align}
V_{S_D}:  &  \ket{i}_{\!S} \ket{x}_{\!B} \to \ket{j(x)+ i }_{\!S}\ket{x}_{\!B} \nonumber 
\end{align}
where $j(x)$ is the index $j$ such that $x \in S_j$.
Here $x \in \mathbb{Z}_N$ and $i$ is an integer between 0 and $n_K-1$, and  the sum $j(x)+i$ is modulo $n_K$ as the first register has dimension  $n_K$.  
This operation  resembles implementation of an oracle with function $j(x)$.
Equivalently using equation \ref{eqn:Pj}
\begin{align}
V_{S_D} & = \sum_{i,j \in \{ 0,..,n_K-1\}} \ket{j+i}_{\!S}\! \bra{i}_S \otimes P_j . \label{eqn:Vj}
\end{align}

As the sets satisfy  $|S_j| = |f(S_j)|$, if we only consider elements of
 a single set $S_j$,  the function $f$ performs part of a permutation
 of ${\mathbb Z}_N$.    If the particular subset $S_j$ is identified such that  $x\in S_j$, 
 then $f(x)$ can be inverted.   For each subset $S_j$ we construct a permutation $\pi_j$
  (a bijective map)
 of ${\mathbb Z}_N$ such that $\pi_i(x) = f(x) $ for $x \in S_j$.  The values of $\pi_j(x)$
 for $x\notin S_j$ must be arranged so that $\pi_j$ is a permutation of the 
elements of ${\mathbb Z}_N$.
The operator 
\begin{align}
U_{f,S} : \ket{j}_{\!S} \ket{x}_{\!B} \to \ket{j}_{\!S} \ket{\pi_j(x)}_{\!B} \label{eqn:Ufs}
\end{align}
is unitary.   The operation can be inverted because $j$ specifies the permutation and the permutation is invertible. 

If we initially have a pure state 
\begin{align}
\ket{\Psi_0} = \ket{0}_{\!S} \sum_x a_x \ket{x}_{\!B}, 
\end{align}
where $a_x$ are complex numbers, 
after marking the sets with the operator in equation \ref{eqn:Vj}
\begin{align}
\ket{\Psi_1}  = \ \sum_x \ket{j(x)}_{\!S} a_x \ket{x}_{\!B} . \label{eqn:psi1}
\end{align}

After applying $U_{f,S}$,  the state
in equation \ref{eqn:psi1} becomes 
\begin{align}
\ket{\Psi_2}  = \ \sum_x \ket{j(x)}_{\!S} a_x \ket{f(x)}_{\!B} \label{eqn:ps2}
\end{align}
The set of operations is summarized in Figure \ref{fig:gen}, with intermediate
states labelled.  
The reset channel ${\cal E}_{\rm reset}$ only resets the bits in the S register.
After reset,  
the state in equation  \ref{eqn:ps2} becomes 
\begin{align}
\rho_{AB}'  & = \ket{0}_{\!S} \!\bra{0}_S \otimes \rho_B'  
\end{align}
with mixed state 
\begin{align}
\rho_B' & =  \sum_j \sum_{x, y \in S_j} 
a_x a_y^* \ket{f(x)}_{\!B} \!\bra{f(y)}_B. 
\end{align} 

The function-generated channel with Kraus
operators given by equation \ref{eqn:ES}  can be performed using the controlled 
operators (based on equation \ref{eqn:Vj}) and unitary operations associated with each disjoint set 
(based on equation \ref{eqn:Ufs}), as shown in Figure \ref{fig:gen}.

\begin{figure}[ht]
\includegraphics[width = 3.3 truein]{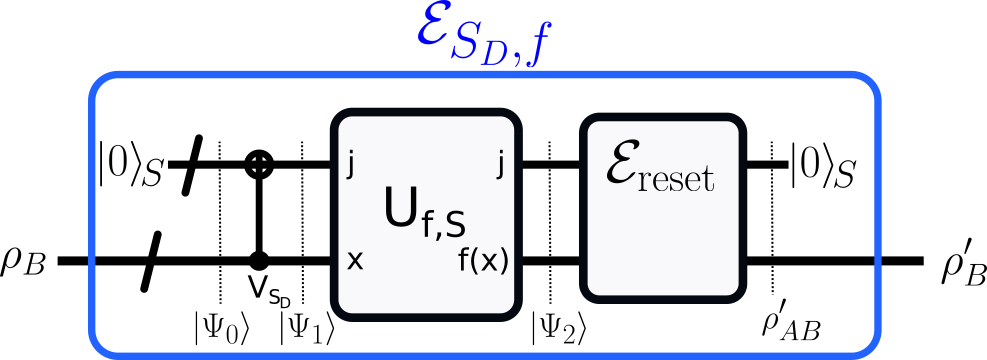}
\caption{This circuit creates a quantum channel ${\cal E}_{S_D,f}$ operating on the bottom register.  The channel is derived from a function $f$ that operates
on ${\mathbb Z}_N$.
The controlled operation on the left  
 (denoted $V_{S_D}$, see equation \ref{eqn:Vj}) and unitary operator
$U_{f,S}$ (equation \ref{eqn:Ufs}) are 
based on a collection of disjoint sets $S_D = \{ S_j\}$ that
satisfy $\cup_j S_j = {\mathbb Z}_N$ and $|f(S_j) = |S_j|$. 
The resulting channel has Kraus operators given in equation \ref{eqn:ES}, as 
defined as in Definition \ref{def:chan}. 
The input of the bottom register is described with density operator $\rho_B$.
The output in the bottom register is $\rho_B' = {\cal E}_{S_D,f}(\rho_B)$.   
The reset operator, ${\cal E}_{\rm reset}$ unentangles the two registers and re-initializes the top register.  The channel makes it possible to iterate the function but is not
necessarily completely dephasing.   For a pure state input, intermediate steps 
$\ket{\Phi_0}, \ket{\Phi_1}$ and $ \ket{\Phi_2}$ are labelled  
 to aid the discussion in the main text. 
 \label{fig:gen}
}
\end{figure}

\subsection{Examples}
\label{sec:examples}

In this subsection we show examples of channels generated from 
a function via Definition \ref{def:chan}.
In subsections \ref{sec:2chan} and \ref{sec:3chan}
we list properties of such all channels derived from functions on 2 and 3 state
systems. In subsection \ref{sec:4chan}  and \ref{sec:echan} we introduce channels on 4 or 8 state systems that preserve coherence
in some subspaces. The logistic map is discussed again in the limit of large $N$ in section \ref{sec:inf}.  The examples illustrate the properties
of this class of channel that are  summarized in section \ref{sec:props}.

\subsubsection{Channels for two state systems}
\label{sec:2chan}

We consider two functions on ${\mathbb Z}_N$ equivalent if they are equivalent after a permutation. 
In other words two functions $f,g$
are equivalent if there is a permutation $\pi(x) \in S_N$ (the permutation group)  such that $\pi^{-1}(f( \pi(x))) = g(x)$. 
We label the channels with a subscript that denotes their Kraus rank.  

Functions on a two state system, equivalent up to permutation,
 include those listed in Table \ref{tab:2el} 
and are illustrated in Figure \ref{fig:class2}.  In Table \ref{tab:2el} we also
note if the function has fixed points $f(x)=x$ or cycles satisfying $f^k(x) = x$. 
Sets of Kraus operators constructed 
from each collection of disjoint sets are listed in Table \ref{tab:2chan}. 
Associated quantum channel properties of each of the entries
in Table \ref{tab:2chan} are listed in Table \ref{tab:2prop}. 
When discussing properties of each channel, we identify the eigenvalues and right and left eigenvectors of the matrix representation of the channel ${\cal L}_{\cal E}$. 
Right eigenvectors with eigenvalues of modulus 1 are chosen so that they are orthogonal (via the Frobenius inner product). 
Left and right eigenvectors with eigenvalues of modulus 1 are chosen to be bi-orthogonal.

On ${\mathbb Z}_2$ there are two ways to choose $f(0)$ and two ways
to choose $f(1)$, so there are four possible functions on ${\mathbb Z}_2$.   Only the two 
constant functions $f(x) = 0$ and $f(x) = 1$ are equivalent up to permutation,  
giving three functions listed in Table \ref{tab:2el}.  
For functions of ${\mathbb Z}_N$  there are $N^N$ possible functions, but many of these are equivalent up to permutation.   The total number of possible functions of ${\mathbb Z}_N$, equivalent up to permutation, is given in appendix \ref{sec:number_f}.
 
Some properties of the class of channel (given via Definition \ref{def:chan})
are evident in Table \ref{tab:2prop}.  If there is a single Kraus operator, the channel
is unitary (channels $B_1$ and $C_1$).  If there are two Kraus operators (with $n_K= N$), the
channel is dephasing (channels $A_2$, $B_2$, $C_2$).  
The possible eigenvalues are $\pm 1$ or 0.  
All the channels have ${\cal E} = {\cal E}^\ddagger$ excepting the $A_2$ channel. 
The $A_2$ channel
illustrates that eigenvectors
need not be perpendicular to each other (via the Frobenius inner product). 
 
\begin{table}[ht]
\caption{Functions on a set of two elements \label{tab:2el}}
\begin{tabular}{llcl}
\hline
type  & function & \# fixed & cycles \\
\hline 
A: & $f(x) =  $  constant  &1&   \\
B: & $f(x) = x$  & 2 &  \\
C: & $f(0) = 1, f(1) = 0$  & 0 &  2 cycle \\
\hline
\end{tabular}
\end{table}

\begin{table}[ht]
\caption{Function-generated channels for 2 state systems \label{tab:2chan}}
\begin{tabular}{lll}
\hline
Channel & Kraus operators & Disjoint sets\\
\hline
$B_1$ & $K_0 =  \ket{0}\!\bra{0} + \ket{1}\!\bra{1} $ & (0,1) \\
$C_1$ & $K_0 =  \ket{1}\!\bra{0} + \ket{0}\!\bra{1} $ & (0,1)\\
$A_2$ & $K_0 = \ket{1}\!\bra{0}, K_1 = \ket{1}\!\bra{1}$& (0),(1) \\
$B_2$ & $K_0 =  \ket{0}\!\bra{0}, K_1 =  \ket{1}\!\bra{1}$& (0),(1)\\
$C_2$ & $K_0 =  \ket{1}\!\bra{0}, K_1 =  \ket{0}\!\bra{1}$& (0),(1)\\
\hline
\end{tabular}
\end{table}

\begin{figure}[ht]\centering
\includegraphics[width=1.5truein]{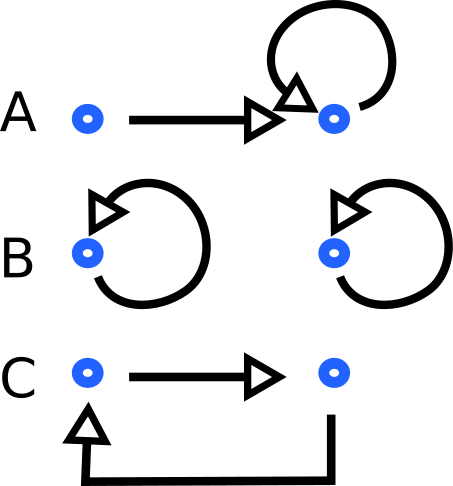}
\caption{Functions of a set of two elements.  
We omit functions that are equivalent via permutation of the elements.  
The blue circles refer to the elements of the set. The black arrows begin at 
an element $x$ and point to an element $f(x)$. 
Only function A is not invertible. 
\label{fig:class2}}
\end{figure}

\begin{table}
\caption{Properties of function-generated channels for 2 state systems \label{tab:2prop}}
\begin{tabular}{llllll}
\hline
Chan. & eigen- & r-eigen- &   l-eigen-  &notes  \\
              & value &  matrix &   matrix   & \\
\hline
$B_1$  & 1 & $\ket{i}\!\bra{j}$ $\forall i,j$ &   ${\cal E} = {\cal E}^\ddagger$  & Identity
       \\
$C_1$ & $\pm 1$ & $\frac{1}{\sqrt{2}}(\ket{0}\!\bra{0} \!\pm\! \ket{1}\!\bra{1}) $ & ${\cal E} = {\cal E}^\ddagger$ &  Unitary    \\
            & $\pm 1$ & $\frac{1}{\sqrt{2}}(\ket{0}\!\bra{1}\! \pm\! \ket{1}\!\bra{0}) $ &   &  Unital   \\
$A_2$ & 1 & $\ket{1}\!\bra{1}$ &   $(\ket{0}\!\bra{0} \!+\! \ket{1}\!\bra{1})$ & \\
           & 0 & $(\ket{1}\!\bra{1}\! -\! \ket{0}\!\bra{0} )$&   $\ket{0}\!\bra{0}$ & eig. not $\perp$   \\
           & 0 & $ \ket{i}\!\bra{j} \  \forall \ i\ne j$ & $ \ket{i}\!\bra{j} \  \forall \ i\ne j$ & Dephasing \\
$B_2$ & 1 & $\ket{0}\!\bra{0}, \ket{1}\!\bra{1}$ &   ${\cal E} = {\cal E}^\ddagger$ & Unital  \\
           & 0 & $ \ket{i}\!\bra{j}\  \forall \ i\ne j$ & & Dephasing \\
$C_2$ & $\pm 1$ & $ \frac{1}{\sqrt{2}}(\ket{0}\!\bra{0}\! \pm \! \ket{1}\!\bra{1})$ & ${\cal E} \!= \!{\cal E}^\ddagger$& 2-cycle, Unital  \\
           & 0 & $ \ket{i}\!\bra{j}\  \forall \ i\ne j$ &   & Dephasing \\
\hline
\end{tabular}
\end{table}

\subsubsection{Channels for three state systems}
\label{sec:3chan}

We consider a qutrit system; a three state system with states $\ket{0}, \ket{1}, \ket{2}$. 
The states are labelled by non-negative integers in ${\mathbb Z}_3$ which is 
the set $\{ 0, 1, 2 \}$. 
Functions on a three state system (equivalent up to permutation)
 include those listed in Table \ref{tab:3el} 
and are illustrated in Figure \ref{fig:class3}. Sets of Kraus operators
for each channel (based on possible collections of disjoint sets) 
are listed in Table \ref{tab:3chan}.  Channel properties such
as its eigenvalues and eigenstates are listed in 
Table \ref{tab:3prop} which we have placed in appendix \ref{ap:3prop}.

Channels (in the class of Definition \ref{def:chan}) for three state system
are somewhat more complex than for two state systems. 
For three state systems, there are generalized right-eigenvectors associated
with eigenvalues that are zero.  In other words, ${\cal L}_{\cal E}$ when 
put in Jordan form, can have non-trivial Jordan blocks associated with zero eigenvalues. 
The possible eigenvalues for the three state channels are 
$\pm 1$, $0$, $\omega$, and $\omega^2$ where $\omega = e^{2 \pi i/3}$. 

\begin{table}[ht]
\caption{Functions on a set of three elements \label{tab:3el}}
\begin{tabular}{llcl}
\hline
type  & function & \# fixed & cycles \\
\hline 
A: & $f(x) =  $  constant  &1&   \\
B: & $f(x) = x$  & 3 &  \\
C: & $f(0) = 1, f(1) = 2, f(2) = 0$  & 0 &  3 cycle \\
D: & $f(0) = 1, f(1) = 0, f(2) = 1$  &  0 & 2 cycle \\
E: & $f(0) = 1, f(1) = 0, f(2) = 2$  & 1  &  2 cycle \\
F: & $f(0) = 0, f(1) = 1, f(2) =0$   & 2   &\\
G: & $f(0) = 0, f(1) = 0, f(2)  =1$  & 1  & \\
\hline
\end{tabular}
\end{table}

\begin{figure}[ht]\centering
\includegraphics[width=2truein]{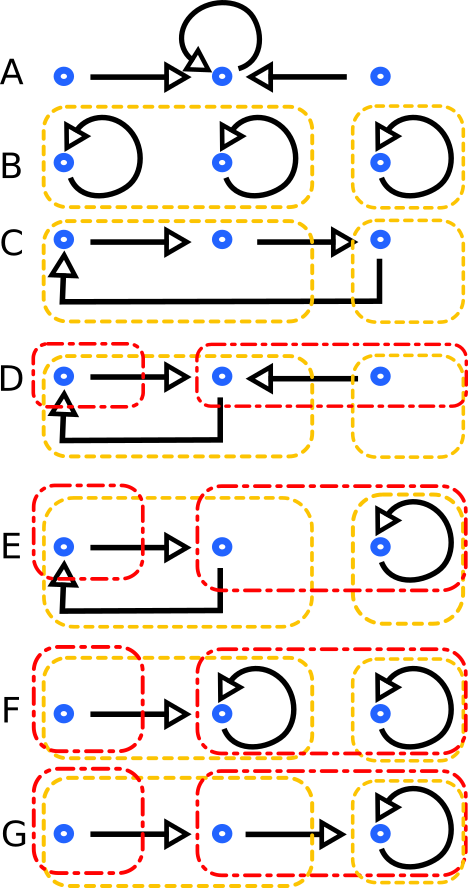}
\caption{Functions of a set of three elements. 
We omit functions that are equivalent after permutation of the elements.  The blue circles refer to the elements of the set. The black arrows begin at an element $x$
and point toward the element $f(x)$. 
$A, D, F$ and $G$ functions are not invertible so there must be at least two Kraus operators in a channel that is generated from these functions.  The orange dashed and red dot-dashed lines show disjoint sets for generating quantum channels with two Kraus operators. The orange dashed 
lines show the disjoint sets for channels denoted $B_2$, $C_2$, $D_{2a}$, $E_{2a}$, $F_{2a}$ and $G_{2a}$ in Table \ref{tab:3el}.
The red dot-dashed 
lines show channels denoted  $D_{2b}$, $E_{2b}$, $F_{2b}$ and $G_{2b}$
in Table \ref{tab:3el}.
\label{fig:class3}}
\end{figure}


\begin{table}[ht]
\caption{Function-generated channels for 3 state systems \label{tab:3chan}}
\begin{tabular}{lll}
\hline
Channel & Kraus operators & Disjoint sets\\
\hline
$B_1$ & $K_0 =  \ket{0}\!\bra{0}+ \ket{1}\!\bra{1} + \ket{2}\!\bra{2}$ &   (0,1,2) \\
$C_1$ & $K_0 = \ket{1}\!\bra{0} + \ket{2}\!\bra{1} + \ket{0}\!\bra{2}$ & (0,1,2)  \\
$E_1$ & $K_0 =  \ket{1}\!\bra{0} + \ket{0}\!\bra{1} + \ket{2}\!\bra{2}$& (0,1,2) \\
\hline
$B_2$     & $K_0 = \ket{0}\!\bra{0} + \ket{1}\!\bra{1},  K_1 = \ket{2}\!\bra{2}$ & (0,1),(2) \\
$C_2$     & $K_0 = \ket{1}\!\bra{0} + \ket{2}\!\bra{1}, K_1 =  \ket{0}\!\bra{2}$& (0,1),(2) \\
$D_{2a}$ & $K_0 = \ket{1}\!\bra{0} + \ket{0}\!\bra{1}, K_1 =  \ket{1}\!\bra{2}$& (0,1),(2) \\
$D_{2b}$ & $K_0 = \ket{0}\!\bra{1} + \ket{1}\!\bra{2}, K_1 =  \ket{1}\!\bra{0}$ & (1,2),(0) \\
$E_{2a}$ & $K_0 = \ket{1}\!\bra{0} + \ket{0}\!\bra{1}, K_1= \ket{2}\!\bra{2}$ &  (0,1),(2) \\
$E_{2b}$ & $ K_0 = \ket{1}\!\bra{0} + \ket{2}\!\bra{2}, K_1= \ket{0}\!\bra{1}$ &  (0,2),(1)  \\
$F_{2a}$ & $K_0 = \ket{1}\!\bra{0} + \ket{2}\!\bra{2}, K_1 = \ket{1}\!\bra{1}$ &  (0,2),(1)\\
$F_{2b}$ & $K_0 = \ket{1}\!\bra{1} + \ket{2}\!\bra{2}, K_1 = \ket{1}\!\bra{0}$&  (1,2),(0) \\
$G_{2a}$ & $K_0 = \ket{1}\!\bra{0} + \ket{2}\!\bra{1}, K_1 = \ket{2}\!\bra{2}$&  (0,1),(2) \\
$G_{2b}$ & $K_0 = \ket{1}\!\bra{0} + \ket{2}\!\bra{2}, K_1 = \ket{2}\!\bra{1}$ &  (0,2),(1)\\
\hline
$A_3$ & $K_0 = \ket{1}\!\bra{0}, K_1 = \ket{1}\!\bra{1}, K_2 = \ket{1}\!\bra{2}$ &(0),(1),(2) \\
$B_3$ & $K_0 = \ket{0}\!\bra{0}, K_1 = \ket{1}\!\bra{1}, K_2 = \ket{2}\!\bra{2}$&(0),(1),(2)\\
$C_3$ & $K_0 = \ket{1}\!\bra{0}, K_1 = \ket{2}\!\bra{1}, K_2 = \ket{0}\!\bra{2}$&(0),(1),(2) \\
$D_3$ & $K_0 = \ket{1}\!\bra{0}, K_1 = \ket{0}\!\bra{1}, K_2 = \ket{1}\!\bra{2}$&(0),(1),(2)\\
$E_3$ & $K_0 = \ket{1}\!\bra{0}, K_1= \ket{0}\!\bra{1}, K_2 = \ket{2}\!\bra{2}$ &(0),(1),(2)\\
$F_3$ & $K_0 = \ket{1}\!\bra{0}, K_1= \ket{1}\!\bra{1}, K_2 = \ket{2}\!\bra{2}$ &(0),(1),(2)\\
$G_3$ & $K_0 = \ket{1}\!\bra{0}, K_1= \ket{2}\!\bra{1}, K_2 = \ket{2}\!\bra{2}$&(0),(1),(2) \\
\hline
\end{tabular}
\end{table}

\subsubsection{Example of a 4 state channel where superposition is maintained in subspaces}
\label{sec:4chan}

\begin{figure}[ht]
\includegraphics[width=1truein]{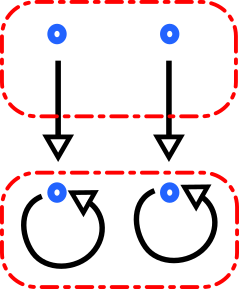}
\caption{An example of a function of 4 elements. 
The red dot-dashed lines show disjoint sets used to construct the Kraus operators for a
particular quantum channel. 
We discuss coherent subspaces in the channel generated from this function 
and these disjoint sets. 
\label{fig:4maps}}
\end{figure}

Consider the example of a function of 4 elements shown in Figure  \ref{fig:4maps}.
We use this function to generate a channel on a 4 state quantum system. 
We assign states so that 
the function has $f(0) = 0, f(1) = 1, f(2) = 0, f(3)=1$. 
We choose disjoint subsets 
$S_D = \{ (2,3), (0,1) \}$ such that one subset contains the chain starting points 2,3, 
and the other subset contains the two fixed points 0,1. 
A set of Kraus operators for this channel along with their associated disjoint subsets 
is
\begin{align}
K_a &= \ket{0}\!\bra{2} + \ket{1}\!\bra{3} \qquad (2,3) \nonumber \\
K_b &=\ket{0}\!\bra{0} + \ket{1}\!\bra{1} \qquad (0,1) . \label{eqn:4chan}
\end{align}
The eigenvalues of this channel are listed in Table \ref{tab:4chan}.
Within the subspace spanned by $\ket{0}, \ket{1}$ the channel
gives the identity transformation.  This is evident from the 4 eigenvectors
that have eigenvalue of 1.  

We consider how the channel operates on a pure and superposition state
$\ket{v} = a \ket{2} + b\ket{3}$ (with $a,b$ complex numbers)
that lies within the subspace spanned by $\{ \ket{2},\ket{3}\} $.
The channel operates on $\ket{v}\!\bra{v}$, as  
${\cal E}(\ket{v}\!\bra{v}) = \ket{w}\!\bra{w}$
with $\ket{w} = a \ket{0 } + b \ket{1}$. 
Note that $\ket{w}$ lies in the subspace spanned by $\{ \ket{0},\ket{1} \} $, 
so subsequent iterations (with $k \ge 1$)  of the channel give 
${\cal E}^k(\ket{v}\!\bra{v}) = \ket{w}\!\bra{w}$.
The superposition is maintained indefinitely.  
A pure superposition state initially in the $\{ \ket{2},\ket{3}\}$ subspace 
will maintain its superposition even after multiple applications of the channel but
 within the asymptotic space spanned by $\{ \ket{0},\ket{1} \}$. 
What happens if the initial state is a superposition $a\ket{2} + b\ket{3} + c \ket{0}$?
After a single application of the channel, the result is a mixed state with density operator 
$ \ket{w}\bra{w} + c c^* \ket{0}\!\bra{0}$ where 
$\ket{w} = a \ket{0} + b \ket{1}$. 
Coherence is maintained for the portion of the initial state in the 
$\{ \ket{2},\ket{3} \} $ subspace.
However, because two subspaces merge into the $\{ \ket{0},\ket{1} \}$ subspace, 
the portion of the state initially in the $\{ \ket{2},\ket{3} \} $ subspace  and that
initially in the $\{ \ket{0}, \ket{1} \} $ subspace are dephased. 

Can we tell from the 
 eigenvalues and eigenvectors of ${\cal L}_{\cal E}$ that superposition of
 a state initially 
 within the $\{ \ket{2},\ket{3} \}$ subspace is maintained by the channel?
 The eigenvectors that do not 
 lie within the $\{\ket{0},\ket{1}\}$ subspace have eigenvalue of zero. This 
 is expected because the asymptotic subspace of the channel is equal to the space
 of density operators of the Hilbert space spanned by $\{\ket{0},\ket{1}\}$.    
In other words, the asymptotic subspace is spanned by $\{ \ket{0}\!\bra{0}$, $\ket{1}\!\bra{1}$, $\ket{0}\!\bra{1}$, $\ket{1}\!\bra{0} \}$. 
Notice that there are four eigenvectors with zero eigenvectors
 that are not perpendicular to matrices in the standard basis $\{ \ket{i}\bra{j} \}$.
Preservation of superposition is related to these non-perpendicular eigenvectors
and to the conserved quantities.   
The conserved quantities (left-eigenvectors with eigenvalue of 1) 
include $\ket{0}\!\bra{0}+\ket{2}\!\bra{2}$  and $\ket{1}\!\bra{1}+\ket{3}\!\bra{3}$. 
The channel sends  $\ket{0}\!\bra{0} \to \ket{2}\!\bra{2}$
so the trace of the subspace spanned by $\ket{0},\ket{2}$ is preserved,
and similarly  the trace of the subspace spanned by $\ket{1},\ket{3}$ is preserved
because ${\cal E}(\ket{1}\!\bra{1}) = \ket{3}\!\bra{3} $.
The conserved quantity  $J  = \ket{0}\!\bra{1}+\ket{2}\!\bra{3} $
and its adjoint arise because off-diagonal operators such as 
$\ket{2}\!\bra{3} \to \ket{0}\!\bra{1}$ via the channel. 
For example, consider how the channel operates on the operator 
$Y = a \ket{3}\!\bra{2} + b \ket{1}\!\bra{0}$ comprised 
of diagonal terms; ${\cal E}(Y) = (a+b) \ket{1}\!\bra{0} $.
This is reflected in the behavior of 
the conserved quantity $J$ gives 
 $\tr (JY) = \tr (J {\cal E}(Y))  = (a+b).$ 

\begin{table}[ht]
\caption{Properties of the channel of Eqn. \ref{eqn:4chan} \label{tab:4chan}}
\begin{tabular}{lll}
\hline
Eigenval. & r-Eigenvec. &  l-Eigenvec. \\
\hline
1 & $\ket{0}\!\bra{0} $ &  $\ket{0}\!\bra{0} \! +\!  \ket{2}\!\bra{2}$\\
1 & $\ket{1}\!\bra{1}$ & $\ket{1}\!\bra{1} \! +\!  \ket{3}\!\bra{3}$ \\
1 & $\ket{0}\!\bra{1} $ & $\ket{1}\!\bra{0} \! +\!  \ket{3}\!\bra{2}$ \\
1 & $\ket{1}\!\bra{0}$ & $\ket{0}\!\bra{1} \! +\!  \ket{2}\!\bra{3}$ \\
0 & $\frac{1}{\sqrt{2}}(\ket{3}\!\bra{3} \! -\!  \ket{1}\!\bra{1})$      & $\ket{3}\!\bra{3}$  \\
0 &  $\frac{1}{\sqrt{2}}(\ket{2}\!\bra{2}\! -\!  \ket{0}\!\bra{0})$      & $\ket{2}\!\bra{2}$ \\
0 &  $\frac{1}{\sqrt{2}}(\ket{2}\!\bra{3}\!  - \! \ket{0}\!\bra{1})$      &  $\ket{3}\!\bra{2}$\\
0 &  $\frac{1}{\sqrt{2}}(\ket{3}\!\bra{2}\!  -\!  \ket{1}\!\bra{0})$      &  $\ket{2}\!\bra{3}$  \\
0 & $\ket{2}\!\bra{1}$ , $\ket{1}\!\bra{2}$&  $\ket{2}\!\bra{1}$ , $\ket{1}\!\bra{2}$  \\
0 & $\ket{3}\!\bra{1}$ , $\ket{1}\!\bra{3}$ &   $\ket{3}\!\bra{1}$ , $\ket{1}\!\bra{3}$  \\
0 & $\ket{2}\!\bra{0}$ , $\ket{0}\!\bra{2}$ &  $\ket{2}\!\bra{0}$ , $\ket{0}\!\bra{2}$ \\
0 & $\ket{3}\!\bra{0}$ , $\ket{0}\!\bra{3}$   & $\ket{3}\!\bra{0}$ , $\ket{0}\!\bra{3}$  \\
\hline
\end{tabular}
\end{table}

\subsubsection{Examples of error correction channels}
\label{sec:echan}

The example discussed in the previous subsection illustrates how a function-generated
channel can preserve coherence in a subspace.  Here 
we illustrate an error correction procedure that can be equivalently described as 
a channel generated from a discrete function.
For example, consider the three-bit bit-flip error correction code  
 where $\ket{0}$ is encoded as $\ket{000}$ and $\ket{1}$ is encoded
as $\ket{111}$.  Using this code, bit flip errors $X_1, X_2, X_3$
(referring to Pauli X gates on the first, second and third bit, respectively)
can be detected and corrected. 
In Figure \ref{fig:ECC}a with arrows, we show the function associated
with the channel that is usually used for error correction. 
Conventionally each type of error is detected (via what is known as a syndrome) 
and then  
corrected, as shown in Figure \ref{fig:ECC}a.   Disjoint sets are shown 
with red dot-dashed loops.  Each correctible error gives a single disjoint set. 
The transformations required for error correction give the function mapping error 
states to encoding states. 
As in our previous example, the disjoint sets
contain more than one set element and this allows coherence to 
be maintained within the different subspaces. 

The function shown in Figure \ref{fig:ECC}b could also be used to correct
the same errors.  With this function,  the channel would need to be called three times to correct all possible bit-flip errors.   A similar type of error correction protocol was
 used by \cite{Sivak_2023}.
 
Consider the circuit of Figure \ref{fig:gen} for implementation of either of these channels.  
The operation denoted $V_{S_D}$ (equation \ref{eqn:Vj}) is equivalent to
a syndrome detection circuit, and the unitary operation $U_{f,S}$ 
(based on equation \ref{eqn:Ufs}) which applies the function, is equivalent
to the error correction operation.  Bits used for syndrome detection must be reset
before the procedure can be iterated, 
otherwise new errors would not be found and corrected with the same syndrome bits.  
As is often true, errors identified
through syndrome detection need not be measured.   The information about
which error has occurred 
is lost when the syndrome detection register is reinitialized. 
Error detection gives a set of disjoint sets where each correctible error defines 
a single set. Error correction gives a discrete map which sends each error state to 
the appropriate state in the encoding subspace. 
Thus error detection and error correction can fit
into the class of quantum channels considered in Definition \ref{def:chan}.

\begin{figure*}[ht]\centering
\includegraphics[width=4.5truein]{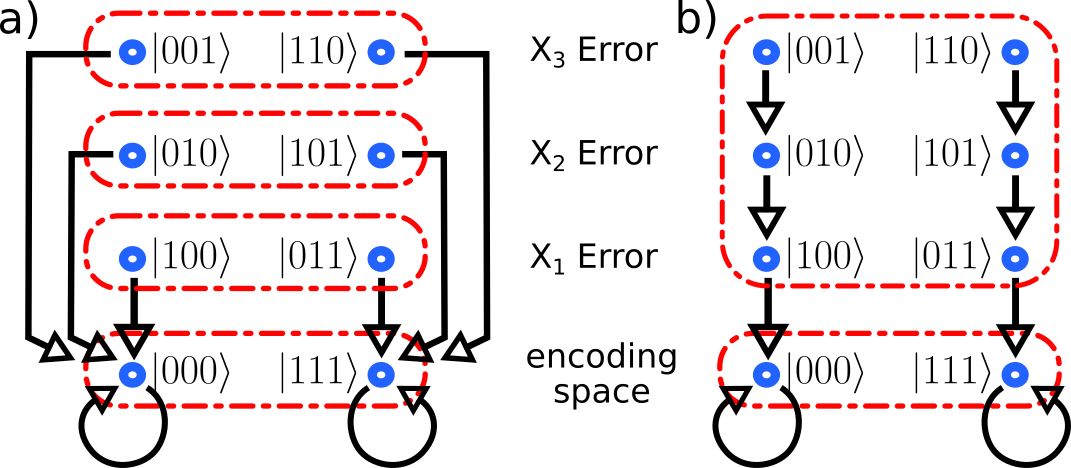}
\caption{a) Error correction for the 3-bit bit-flip error correcting code
is equivalent to a channel that can be described as generated from a discrete function. 
The red dot-dashed loops show the disjoint sets used to generate the Kraus operators
for the channel. 
b) Error correction can also be done with a different discrete function
that uses fewer Kraus operators. 
\label{fig:ECC}}
\end{figure*}



\subsubsection{An example of a channel on an infinite system}
\label{sec:inf}

Can function-generated channels (via Definition \ref{def:chan}) be extended in the limit of the number of states $N\to \infty$?
Consider the continuous set of states on the Hilbert space of the unit interval $\ket{x}$ with $x \in [0,1)$
which would be the infinite dimensional limit of the example described
in section \ref{sec:trunc} of the truncated logistic map. 
The peak of the logistic map $g(x) = \mu x(1-x)$ (previously equation 
\ref{eqn:logistic}) is at $x=1/2$. 
We divide the unit interval into two segments $0 \le x < 1/2$ and $1/2 \le x <1$.
Given $g(x)$ and identifying the segment that contains $x$, 
 we can invert $g(x)$. 
The two segments are two disjoint sets $S_0 = [0,1/2), S_1 = [1/2,1)$ 
and they can be used to create a channel that applies 
 the logistic map.   
The channel is formed from only two Kraus operators
\begin{align}
K_0 &= \int_0^\frac{1}{2} dx \ket{g(x)}\bra{x}\nonumber \\
K_1 & = \int_{\frac{1}{2}}^1 dx \ket{g(x)}\bra{x}.
\end{align}
This channel can have chaotic orbits as the logistic map itself 
has them.   

A channel for an infinite dimensional quantum system may be difficult to implement.  
Nevertheless, as a thought experiment we can  
look at the operators for implementing the channel, 
following section \ref{sec:realize}. 
We would only need a single additional qubit in the S register to denote the segment.  
The B subspace would be the Hilbert space of the unit interval. 
The projection operators (via equation \ref{eqn:Pj}) on the B subspace would be 
\begin{align}
P_0 &=  \int_0^\frac{1}{2} dx \ket{x}\!\bra{x}\nonumber \\
P_1 & = \int_{\frac{1}{2}}^1 dx \ket{x}\!\bra{x}.
\end{align}
The controlled operator operating on the bipartite 
quantum system (via equation \ref{eqn:Vj}) 
\begin{align}
V_{S_D} &= ( \ket{0}_{\!S}\!\bra{0}_S  + \ket{1}_{\!S}\!\bra{1}_S) \otimes \int_0^\frac{1}{2} dx  \ket{x}_{\!B}\!\bra{x}_B  \nonumber \\
 & \ \ \    
+  (\ket{0}_{\!S}\!\bra{1}_S +  \ket{1}_{\!S}\!\bra{0}_S)  ) \otimes\int_{\frac{1}{2}}^1 dx  \ket{x}_{\!B} \!\bra{x}_B .
\end{align}

The associated unitary operation (via equation \ref{eqn:Ufs}) 
could be 
\begin{align}
U_{f,S}: &  \ket{0}_{\!S} \ket{x}_{\!B} \to \begin{cases}
\ket{0}_{\!S}\ket{g(x)}_{\!B} & {\rm for\ } 0<x \le \frac{1}{2} \\
\ket{0}_{\!S}\ket{v(x)}_{\!B}& {\rm for\ } x >\frac{1}{2} \nonumber \\
\end{cases}\\
& \ket{1}_{\!S} \ket{x}_{\!B} \to \begin{cases}
\ket{1}_{\!S}\ket{v(x)}_{\!B} & {\rm for\ } 0 <x \le \frac{1}{2} \\
\ket{1}_{\!S}\ket{g(x)}_{\!B}& {\rm for\ } x >\frac{1}{2} \\
\end{cases}. \label{eqn:U_log}
\end{align}
where $v(x)$ is a V-shaped function on the unit interval 
that has $v(0) = 1, v(1) =1$
and $v(1/2) = g(1/2)$, as shown in Figure \ref{fig:U_log}.
\begin{figure}
\includegraphics[width=1.5 truein]{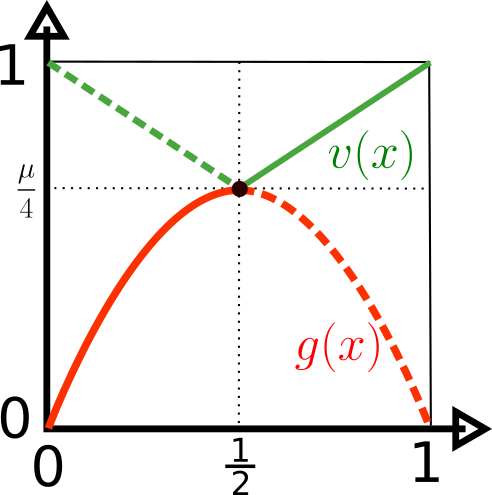}
\caption{We show a V-shaped function $v(x)$ (in green) along with 
the logistic map $g(x) = \mu x (1-x)$ (in red).  Two invertible functions
 can be constructed from $v()$ and $g()$ using left and
 right sub-intervals of the unit interval with dividing point at $x=1/2$.
One invertible function (shown as solid lines) is equal to $g(x)$ for $x \in[0, 1/2)$ and $v(x)$ for $x\in[1/2,1)$.
The other (shown with dashed lines) is equal to $v(x)$ for $x \in[0, 1/2)$ and is $g(x)$ for $x\in[1/2,1)$.  
These two invertible functions are used to create the unitary operation
 in equation \ref{eqn:U_log}.
\label{fig:U_log}}
\end{figure}

Because only a single additional qubit is needed to create the channel,  
 only a small amount of information need be removed during each iteration
of the channel.  Despite the relatively low information loss  
 with only two Kraus operators, the associated operation 
is usually dephasing as orbits (including the attracting periodic orbits)
tend to go back and forth across the midpoint $x=1/2$.   
The channel, when iterated, would mostly behave classically.


\subsection{Properties of channels generated from a function}
\label{sec:props}

Examination of the examples given in the previous sections 
helps us identify characteristics of the quantum channels 
generated from a function as defined in Definition \ref{def:chan}.
 
Properties of a channel ${\cal E}_{S_D,f}$ defined via a function $f$ and a set of disjoint sets $S_D$ (equation \ref{eqn:ES}) in Definition \ref{def:chan}
\begin{enumerate}
\item
The minimum number of Kraus operators describing a channel associated
with $f$ is equal to 
 $N-|f(\mathbb{Z}_N)|+1 $,  which the smallest possible number of allowed disjoint sets 
arising from the completeness requirement for the set of Kraus operators. 
\item
The maximum number of Kraus operators describing the channel is equal to $N$.
This is much smaller than the maximum
possible number of Kraus operators for any quantum channel ($N^2$).
This type of channel is a small subset of all possible quantum channels. 
\item 
If the Kraus rank of the channel  is 
the same as the dimension of the quantum space, $n_K = N$, then the channel is dephasing.  
This follows as each Kraus operator performs a permutation of the set elements
and all off diagonal elements must be eigenvectors of the matrix representation
of the channel with eigenvalue of zero.  
\item 
If $f$ is bijective and there is a single Kraus operator in the channel,
then the channel is a unitary transformation that consists of a permutation of basis states. 
\item 
Eigenvalues of the matrix representation of the channel ${\cal L}_{\cal E}$ are either roots of unity or zero. 
For the outline of a proof see appendix \ref{ap:proofs}. 
The asymptotic subspace for these channels is spanned by the right eigenvectors 
of  ${\cal L}_{\cal E}$  with non-zero eigenvalues as these all have modulus 1.   
\item
If a collection of disjoint sets  $\{ S_j \}$ has sets that contain more than 1 element, 
then a different channel can be constructed by dividing one or more sets into collections of smaller sets. 
The resulting channel has more right-eigenvectors or/and generalized right-eigenvectors (of ${\cal L}_{\cal E}$) with zero eigenvalues than to its parent.  (See the appendix  \ref{ap:proofs} for how these eigenvalues 
and eigenvectors arise.)  We give some examples.  The $D_{2a}$ channel 
has 5 non-zero eigenvalues.   When split to become the $D_3$ channel,
there are two additional eigenvectors with eigenvalues of zero.  
The $D_{2b}$ channel has 5 zero eigenvalues
and 2 generalized eigenvectors associated with 0 eigenvalues. 
When split to become the $D_3$ channel,  the generalized eigenvectors  in the $D_{2b}$ channel become regular eigenvectors. 
\item
Each $k$-periodic orbit of the function $f$, satisfying $f^k(x) = x$, 
generates a $k$-periodic orbit of the channel.  The orbit has initial condition  
$\rho = \ket{x}\bra{x}$ and satisfies ${\cal E}^k (\rho) = \rho$. 
The states in the cycle $\ket{f^j(x)}\bra{f^j(x)}$ for $j \in \{0, 1, ..., k-1\}$ 
are density operators with trace of 1.   
The states in the cycle are orthogonal with respect to the 
Hilbert-Schmidt inner product. 
\item 
A $k$-periodic orbit in $f$ generates a fixed point of the channel that is
\begin{equation}
\rho_{*,k} = \frac{1}{k} \sum_{j=0}^{k-1} \ket{f^j(x)}\bra{f^j(x)}.
\label{eqn:rhok}
\end{equation}
This fixed point is a valid density matrix as its trace is 1 and 
it is a mixture of pure states. 
\item 
A $k$-periodic orbit of the function $f$, satisfying $f^k(x) = x$, 
generates an invariant subspace (with respect to ${\cal E}_{S_D,f}$)  in the space of density matrices.
The invariant subspace is the set of density matrices
$S_x  = \{ \rho: \rho = \sum_{j=0}^{k-1} a_j \ket{f^j(x)} \bra{f^j(x)}\}$ with the set $\{a_j\}$ probabilities; $\sum_{j=0}^{k-1} a_j = 1$.  The subspace satisfies 
${\cal E}_{S_D,f}(\rho) \in S_x$,  for any $\rho \in S_x$, 
The set $S_x$ contains the fixed point $\rho_{*,k}$ defined in equation \ref{eqn:rhok}.
All $\rho \in S_x, \rho \ne \rho_{*,k}$ generate $k$-periodic orbits that remain within $S_x$. 
\item 
Each $x$, satisfying $f^k(x) = x$, generates $k$ right eigenvectors (of ${\cal L}_{\cal E}$)
in the form 
\begin{equation}
\rho_m = \frac{1}{k} \sum_{j=0}^k \omega^{mj} \ket{f^j(x)}\!\bra{f^j(x)},  \label{eqn:rhom}
\end{equation}
satisfying ${\cal E}(\rho_m) = \omega^{-m} \rho_m$, 
with $m\in \{ 0, 1,..., k-1\}$ and $\omega = e^{2 \pi i/k}$.  
These eigenvectors 
have eigenvalues that are complex roots of unity.  
The eigenvalue of eigenvector $\rho_m$ is equal to $\omega^{-m}$. 

\item
Orbits of the channel either decay to zero, to a fixed point, or enter a cycle.
We refer to these as end-states. 
This follows from the possible eigenvalues.   
The maximum number of iterations for the channel to achieve an end state, 
from any initial condition, 
is equal to the maximum link length $k_c$ of the function $f$ (with $k_c$ defined in equation \ref{eqn:k_c}).
For how this arises, see appendix \ref{ap:proofs}. 
\item
How many zero eigenvalues of ${\cal L}_{\cal E} $ are present?   
Let $n_f = | f({\mathbb Z}_N)|$
be the number of elements in the range of the function. 
In the matrix basis, the range of the channel is at most  
$n_f^2$.  This implies that there are at least $N^2 - n_f^2$ eigenvectors 
or generalized eigenvectors  associated with zero eigenvalues. 
\item If all members of a $k$-cycle of $f$ are members of a subset 
$S_j$, i. e., 
$f^m(x) \in S_j$ for $m=0, 1, ...., k-1$, 
then the channel gives a unitary transformation of
the vector subspace in $\cal H$ generated via $\{ \ket{f^m(x)} \}$ with $m = 0, 1, ...., k-1$. 
The unitary transformation is equivalent to a cyclic permutation 
of the basis states of this vector subspace. 
\item 
For two integers in a disjoint set, $x,y \in S_j$,  with $S_j \in S_D$, a single operation
of the channel preserves
a superposition in the subspace spanned by $\ket{x}, \ket{y}$.
In other words if we take $\ket{v} = a \ket{x} + b \ket{y}$,  ($a, b$ complex numbers) 
the channel 
gives ${\cal E} (\ket{v}\bra{v}) = \ket{w}\bra{w}$ with 
$\ket{w} = a \ket{f(x)} +  b \ket{f(y)}$. 
This is behavior is described in more detail in section \ref{sec:4chan}
and is a necessary component of error correction channels.
\end{enumerate}

\section{Summary and Discussion}
\label{sec:sum}

We construct a quantum channel using a function on a discrete set 
and the ability to reset or initialize a quantum register. 
The channel makes it possible to iterate a non-invertible oracle function. 
With oracle function 
 a truncated version of the logistic map, we construct a quantum channel that  
exhibits much of the phenomenology of the logistic map, including
attracting periodic orbits.   The quantum channel is dephasing, so its iteration is 
 essentially a classical operation.  We used the logistic map because its computation only
  requires addition and multiplication operations. The channel 
 can exhibit attracting periodic orbits of any period, depending upon the size
 of the quantum system and the single adjustable parameter for the map. 
This simple construction illustrates a way to construct quantum channels
that are not ergodic, in the sense that they can exhibit more than one stationary 
state (a fixed point) in the space of density operators and the channel 
can contains orbits that are non-decaying cycles. 
Function-generated channels give 
 straightforward examples of quantum channels that display the phenomena 
of cycles that can be present in the asymptotic limits of some quantum channels 
\cite{Wolf_2012,Albert_2019,Carbone_2020,Amato_2022}. 

A possible application of dephasing quantum channels that iterate functions is
in hybrid quantum/classical algorithms \citep{Callison_2022}.  
Using a unitary quantum oracle, 
swap functions and a reset/initialization of one quantum register, we constructed 
a quantum channel that, when iterated, computes the greatest common divisor of two integers.    While it seems amusing to employ the 2300 year old Euclidian algorithm
on a quantum computer, this algorithm is a component in Shor's factoring algorithm. 
Euclid's algorithm for finding the greatest common divisor does not require many operations in comparison to the quantum components of Shor's algorithm itself, 
so if  transferring information back and forth (via state
initialization and measurement) between
quantum and classical computers is difficult, slow or error prone, 
then this classical component of
Shor's algorithm could be done in situ on a quantum computer. 

It has long been known that any classical operation can be carried out on a quantum computer (e.g., \cite{Nielsen_2010, Rieffel_2011}).  However, 
when non-invertible functions are iterated, an increasing number 
of ancillary qubits are necessarily entangled (e.g., \cite{Seidel_2022}).  
By resetting registers,  we have sacrificed 
coherence for the ability to reuse ancillary qubits and this gives the ability
to iterate non-invertible functions.  

Function-generated channels are not limited to completely dephasing channels.
We illustrate with Definition \ref{def:chan} 
how channels can be constructed from discrete functions using
disjoint subsets and 
 with fewer Kraus operators than in the completely dephasing case. 
There is a relation between the size of 
the discrete function's range and the possible choices for the set of 
 Kraus operators determining the channel. 
This choice is related to the loss of information that occurs 
when a non-invertible function is computed and its original arguments
are forgotten. 
With more Kraus operators, more information is lost during application 
 of the channel, resulting in additional eigenvectors with
 zero eigenvalues and a reduction in the dimension or number of coherent subspaces.
 
The function-generated channels described here include those that exhibit 
 coherent subspaces, similar to some error correction routines. 
 In this setting, 
syndrome detection is equivalent to marking each state with an associated disjoint set 
 and error correction is performed by applying a unitary
operation that permutes the basis element and that depends on the disjoint set.
As syndrome detection must be done with initialized states, the procedure
of syndrome detection followed by error correction, when 
iterated, is equivalent to a iterating a quantum channel that 
 initializes or resets the state within a subspace and then applies a unitary transformation. 
 
Future applications of the simple channels discussed here could include the development
of novel random walks on open quantum systems \cite{Carbone_2020}, 
 generating additional types of POVMs \cite{Buscemi_2005}, varying
 procedures for error correction \citep{Sivak_2023} and computing components of 
optimization routines that are typically done classically, but in situ
on a quantum computer within quantum/classical hybrid algorithms.  



\begin{acknowledgments}

This work was inspired by discussions with Ray Parker. 
We thank Ray Parker, Joey Smiga, Liz Champion, Andrew Jordan, Phillip Lewalle, Gil Rivlis,  Gabriel Landi, Machiel Bloch, and Max Neiderbach for helpful discussions. 
We thank Ray Parker and Joey Smiga for insightful comments on this manuscript.  

\end{acknowledgments}

{\bf Data Availability Statement} 

All data generated or analysed during this study are included in this published article.





\bibliography{circ}

\appendix

\section{Number of functions on $\mathbb{Z}_n$ up to permutation}
\label{sec:number_f}

The number of possible functions on 
a discrete set, equivalent up to permutation, determines the potential range of phenomena exhibited by the class of quantum channel defined in Definition \ref{def:chan}.  Using work by \cite{Shattuck-Holloway}, we find that we can give a expression for the total number of possible functions of  $\mathbb{Z}_n$ in the equivalence class.  

Let $n\in\mathbb{Z}^+$, and let $X$ denote the set of functions from $\mathbb{Z}_n$ to itself. The permutation group on $\mathbb{Z}_n$, which we denote by $S_n$, acts on $X$ via conjugation:
\begin{equation}
    \sigma\circ f\circ \sigma^{-1},
\end{equation}
with $\sigma \in S_n$. 
If $f\in X$, then the orbit of $f$ under this action is the set:
\begin{equation}
    orb(f)=\{g\in X:\exists \sigma\in S_n\text{ s.t }g=\sigma\circ f\circ\sigma^{-1}\}.
\end{equation}
If $f,g\in X$, we say $f$ is equivalent to $g$ up to permutation, and write $f\sim g$, if and only if there is a  $\sigma\in S_n$ such that $g=\sigma\circ f\circ\sigma^{-1}$. The relation $\sim$ is an equivalence relation.
Moreover, the orbit of any $f\in X$ is precisely the set of all functions equivalent to $f$ up to permutation. We count the number of functions from $\mathbb{Z}_n$ to itself considered up to permutation equivalence, i.e, count the total number of distinct orbits under this group action. 

For a fixed $\sigma\in S_n$, we define
\begin{equation}
X_\sigma=\{f\in X:\sigma\circ f\circ\sigma^{-1}=f\} \label{eqn:X_sig}
\end{equation}
which is the set of all functions from $\mathbb{Z}_n$ to itself which commute with permutation $\sigma$. 

Any permutation $\sigma\in S_n$ can be written uniquely as a product of disjoint cycles (of possibly differing lengths). Keeping with the notation of \cite{Shattuck-Holloway}, we let $\lambda_i$ denote the number of cycles of length $i$ for $\sigma$ written in cyclic notation in terms of disjoint cycles. Proposition 2.6 from \cite{Shattuck-Holloway} states that the number $N_{\rm pairs}$ of ordered pairs $(\tau,f)$, with permutation $\tau \in S_n$, function $f \in X$, and such that 
the permutation commutes with the function; $\tau\circ f=f\circ \tau$, is
\begin{align}
N_{\rm pairs}  =
  & \!\!\!\!\!\ \sum\limits_{\substack{1\lambda_1+...+n\lambda_n=n, \\ \:\:\lambda_i\geq 0}}\frac{n!}{\lambda_1!\cdot...\cdot\lambda_n!1^{\lambda_1}\cdot...\cdot n^{\lambda_n}}\nonumber  \\
  & \qquad \qquad \times 
  \prod_{k=1}^{n}\left(\sum_{j|k}j\lambda_j\right)^{\lambda_k}.
  \label{eqn:N_pairs}
\end{align}
The form of proposition 2.6 stated by  \cite{Shattuck-Holloway} is defined in terms of functions on  $\{1,...,n\}$, however,  the version of Proposition 2.6 stated above is equivalent as $\{1,...,n\}$ and $\mathbb{Z}_n$ are finite sets of the same size. Using this proposition, we find the following:
\begin{theorem}
    Let $n\in\mathbb{Z}^+$. The number of functions from $\mathbb{Z}_n$ to itself, considered up to permutation equivalence, is
    \begin{equation}
    \sum\limits_{\substack{1\lambda_1+...+n\lambda_n=n,\\ \:\:\lambda_i\geq 0}}\frac{1}{\lambda_1!\cdot...\cdot\lambda_n!1^{\lambda_1}\cdot...\cdot n^{\lambda_n}}\prod_{k=1}^{n}\left(\sum_{j|k}j\lambda_j\right)^{\lambda_k}
    \label{eqn:nresult}
    \end{equation}
\end{theorem}
\begin{proof}
We denote $X$ to be the set of all functions from $\mathbb{Z}_n$ to itself.
    The number of functions $f \in X$,  equivalent up to permutation, is equal to the number $r$ of distinct orbits under the conjugation action of the permutation group $S_n$ on $X$.  This is a setting where we can apply Burnside's lemma for the number of orbits, giving 
    \begin{equation}
        r=\frac{1}{n!}\sum_{\sigma\in S_{\mathbb{Z}_n}}|X_{\sigma}|.
        \label{eqn:r}
    \end{equation}
    Here we have used the fact that $n! = |S_n|$ and $X_{\sigma}$ as defined in equation \ref{eqn:X_sig}.
    The sum $\sum_{\sigma\in S_n}|X_{\sigma}|$ is precisely the number of all ordered pairs $(\tau,f)$ such that $\tau\circ f=f\circ \tau$, where permutation $\tau \in S_n$ and $f$ is a function from $\mathbb{Z}_n$ to itself. Applying Proposition 2.6 from \cite{Shattuck-Holloway} (via equation \ref{eqn:N_pairs}) with equation \ref{eqn:r}, we obtain the desired result of equation \ref{eqn:nresult}.
\end{proof}

\section{Properties of function-generated channels for 3-state systems}
\label{ap:3prop}

In Table \ref{tab:3prop} we include 
the eigenvalues and left and right eigenvectors of the channels for
three state systems 
with Kraus operators listed in Table \ref{tab:3chan} and illustrated in Figure 
\ref{fig:class3}. By r-eigenvector we mean right eigenvector and by l-eigenvector
we mean left eigenvector. 
We list left eigenvectors if the dual channel
differs from the channel;  $\cal E$ differs from ${\cal E}^\ddagger$. 
For eigenvalues of modulus 1, we choose orthogonal right eigenvectors, and conserved quantities or left eigenvectors  that satisfy a biorthogonality condition via the Frobenius inner product. 
These examples help illustrate different cases in the proof of appendix  
\ref{ap:proofs}. 

\begin{table*}[ht]\footnotesize
\caption{Properties of function-generated channels for 3 state systems \label{tab:3prop}}
\begin{tabular}{lllll}
\hline
Channel & eigenvalue & r-eigenvec. &  l-eigenvec. & notes \\
\hline
$B_1$  & 1 & $\ket{i}\bra{j}$ $\forall i,j$ & 
${\cal E} = {\cal E}^\ddagger$ & Identity channel  \\
$C_1$  & $\omega^k$ $k \in {\mathbb Z}_3$ & $\frac{1}{\sqrt{3}}\sum_{j=0}^2 \omega^{-jk} \ket{j}\!\bra{j}$   &   ${\cal E} \sim  {\cal E}^\ddagger$  via permutation & Unitary,  $\omega=e^{2\pi i/3}$\\
  & $\omega^k$ & $\frac{1}{\sqrt{3}}\sum_j \omega^{-jk} \ket{j}\!\bra{j+1}$ &&  Unital, 3-cycles  \\
  & $\omega^k$ & $\frac{1}{\sqrt{3}}\sum_j \omega^{-jk} \ket{j+1}\!\bra{j}$ & \\
$E_1$  & $\pm 1$ & $\frac{1}{\sqrt{2}}(\ket{0}\!\bra{0} \pm \ket{1}\!\bra{1}) $&  ${\cal E} = {\cal E}^\ddagger$ & Unitary, 2-cycles \\ 
	    & $\pm 1$ & $\frac{1}{\sqrt{2}}(\ket{0}\!\bra{1}\! \pm\!  \ket{1}\!\bra{0}) $& & Unital   \\
	    & $\pm 1$ & $\frac{1}{\sqrt{2}}(\ket{0}\!\bra{2} \! \pm \! \ket{1}\!\bra{2}) $& \\
	    & $\pm 1$ & $\frac{1}{\sqrt{2}}(\ket{2}\!\bra{0} \! \pm\!  \ket{2}\!\bra{1}) $& \\
            & 1       & $\ket{2}\bra{2} $ & \\
\hline
$B_2$  & 1 & $\ket{i}\!\bra{i}$ $\forall i$ &${\cal E} = {\cal E}^\ddagger$& Identity in subspace $\{ \ket{0},\ket{1}\} $   \\
             &  1  & $\ket{1}\!\bra{0}, \ket{0}\!\bra{1}$ &  & Unital   \\
             &  0 & $ \ket{0}\!\bra{2}, \ket{2}\!\bra{0}, \ket{2}\!\bra{1},\ket{1}\!\bra{2} $ & \\
$C_2$  & $\omega^k$ & $\frac{1}{\sqrt{3}}\sum_j \omega^{-jk} \ket{j}\bra{j}$  & 
                ${\cal E} \sim  {\cal E}^\ddagger$  via permutation & Unital, 3-cycle \\
             & 0  & $\ket{i}\!\bra{2}, \ket{2}\!\bra{i} $ for $i \in \{0,1\}$ & \\
             & 0-associated & $\ket{0}\!\bra{1}, \ket{1}\!\bra{0}$   & & generalized eigenvectors \\
$D_{2a}$ & $\pm 1$  & $\frac{1}{\sqrt{2}}(\ket{0}\!\bra{0}\! \pm\! \ket{1}\!\bra{1})$ &  $\ket{0}\!\bra{0} \!\pm\! \ket{1}\!\bra{1} \!+\! \ket{2}\!\bra{2}$& Unitary in subspace $\{ \ket{0},\ket{1}\}$  \\
               & $\pm 1$  & $\frac{1}{\sqrt{2}}(\ket{0}\!\bra{1}\! \pm\! \ket{1}\!\bra{0})$ & same & 2-cycle    \\
                & 0  & $\ket{0}\!\bra{2}, \ket{2}\!\bra{0}, \ket{1}\!\bra{2}, \ket{2}\!\bra{1}$ & same &\\
                & 0 & $ \frac{1}{\sqrt{2}}(\ket{2}\!\bra{2} \!-\!\ket{0}\!\bra{0}) $  not $\perp$ &  $\ket{2}\!\bra{2}$  & \\
$D_{2b}$ & $\pm 1$  & $\frac{1}{\sqrt{2}}(\ket{0}\!\bra{0}\! \pm\! \ket{1}\!\bra{1})$ & 
               $\frac{1}{\sqrt{2}}(\ket{0}\!\bra{0} \!\pm\! \ket{1}\!\bra{1} \!+\! \ket{2}\!\bra{2})$  & 2-cycle \\
                 & 0 &  $\ket{0}\!\bra{1}, \ket{1}\!\bra{0} , \ket{0}\!\bra{2}, \ket{2}\!\bra{0}$ & 
                      $\ket{0}\!\bra{1}, \ket{1}\!\bra{0} , \ket{1}\!\bra{2}, \ket{2}\!\bra{1}$ & \\
                & 0   & $\frac{1}{\sqrt{2}}(\ket{2}\!\bra{2}\! -\! \ket{0}\!\bra{0})$  not $\perp$ & $\ket{2}\!\bra{2} $ \\
                & 0-associated & $\ket{1}\!\bra{2}, \ket{2}\!\bra{1}$ & 
                $\ket{0}\!\bra{1}, \ket{1}\!\bra{0}$
                & generalized eigenvectors \\
 $E_{2a}$  & $\pm 1$ & $\frac{1}{\sqrt{2}}(\ket{0}\!\bra{0}\! \pm \!\ket{1}\!\bra{1})$ &${\cal E} = {\cal E}^\ddagger$ & Unital, 2-cycle \\
                  & $\pm 1$ & $\frac{1}{\sqrt{2}}(\ket{0}\!\bra{1}\! \pm \!\ket{1}\!\bra{0})$ & & Unitary in subspace $\{ \ket{0}, \ket{1}\} $  \\
                  & 1 & $\ket{2}\!\bra{2}$ &  \\
                  & 0 & $\ket{0}\!\bra{2}, \ket{2}\!\bra{0}, \ket{1}\!\bra{2},\ket{2}\!\bra{1}$& \\
 $E_{2b}$ &   $\pm 1$  & $\frac{1}{\sqrt{2}}(\ket{0}\!\bra{0}\! \pm\! \ket{1}\!\bra{1} )$ & same & Unital, 2-cycle\\
                 &   1  &  $\ket{2}\bra{2}  $  &  same &\\
                 & 0 & $\ket{0}\!\bra{1}, \ket{1}\!\bra{0}, \ket{1}\!\bra{2},\ket{2}\!\bra{1} $ & 
                       $\ket{0}\!\bra{1}, \ket{1}\!\bra{0}, \ket{0}\!\bra{2},\ket{2}\!\bra{0} $ \\
                 & 0-associated & $\ket{0}\!\bra{2},\ket{2}\!\bra{0}$ & $\ket{1}\!\bra{2},\ket{2}\!\bra{1} $& generalized eigenvectors\\
$F_{2a}$ &  1 & $\ket{1}\!\bra{1}, \ket{2}\!\bra{2}$ &  $\ket{2}\!\bra{2}$, $\ket{0}\!\bra{0} \!+ \!\ket{1}\!\bra{1}$  & \\
                & 0  & $\ket{1}\!\bra{0}, \ket{0}\!\bra{1}, \ket{1}\!\bra{2}, \ket{2}\!\bra{1}$ & same & \\
                 & 0 & $\frac{1}{\sqrt{2}}(\ket{0}\!\bra{0}\! - \!\ket{1}\!\bra{1})$ not $\perp$ & $\ket{0}\!\bra{0}$ \\
                 & 0-associated & $\ket{0}\!\bra{2}, \ket{2}\!\bra{0}$ & $\ket{2}\!\bra{1},\ket{1}\!\bra{2}$ & generalized eigenvectors \\
$F_{2b}$ &  1 & $\ket{2}\!\bra{2},\ket{1}\!\bra{2}, \ket{2}\!\bra{1} $ &  same 
& Identity in subspace $\{ \ket{1},\ket{2} \}$ \\
&  1 & $\ket{1}\!\bra{1}$ & $\ket{0}\!\bra{0}\!+\! \ket{1}\!\bra{1}$  & \\
               &  0 &   $\ket{0}\!\bra{1}, \ket{1}\!\bra{0}, \ket{0}\!\bra{2},\ket{2}\!\bra{0}$   & same & \\
               &  0  &$\frac{1}{\sqrt{2}}(\ket{0}\!\bra{0} \!-\! \ket{1}\!\bra{1})$   not $\perp $ & $\ket{0}\!\bra{0} $& \\
$G_{2a}$ &  1 & $\ket{2}\!\bra{2} $ & $\ket{0}\!\bra{0}\!+\!\ket{1}\!\bra{1}\!+\!\ket{2}\!\bra{2}$ &Ergodic, mixing\\
                & 0  & $\ket{0}\!\bra{2}, \ket{2}\!\bra{0}, \ket{2}\!\bra{1}, \ket{1}\!\bra{2}$ & $\ket{0}\!\bra{2}, \ket{2}\!\bra{0},\ket{0}\!\bra{1}, \ket{1}\!\bra{0}$&\\
                &  0 & $ \frac{1}{\sqrt{2}}(\ket{1}\!\bra{1} \!-\! \ket{2}\!\bra{2} )$ not $\perp$ & $\ket{0}\!\bra{0}$ \\
                & 0-associated & $\frac{1}{\sqrt{2}}( \ket{0}\!\bra{0} \!-\! \ket{2}\!\bra{2})$ & $ \ket{0}\!\bra{0} $  & generalized eigenvector \\
                & 0-associated & $\ket{1}\!\bra{0}, \ket{1}\!\bra{0}$ &$\ket{0}\!\bra{0}\!+\! \ket{1}\!\bra{1}$ & generalized eigenvectors \\
$G_{2b}$ &  1 & $\ket{2}\!\bra{2}$ & $\ket{0}\!\bra{0}\!+\! \ket{1}\!\bra{1}\!+\!\ket{2}\!\bra{2}$& Ergodic, Mixing \\
                 & 0 & $ \ket{1}\!\bra{2}, \ket{2}\!\bra{1}, \ket{0}\!\bra{1},\ket{1}\!\bra{0}$ &
                         $ \ket{0}\!\bra{1}, \ket{1}\!\bra{0}, \ket{0}\!\bra{2},\ket{2}\!\bra{0}$ \\
	        & 0 & $\frac{1}{\sqrt{3}}(\ket{0}\!\bra{0}\! +\!\ket{1}\!\bra{1}\!-\! \ket{2}\!\bra{2})$  not $\perp$ & $\ket{0}\!\bra{0}$ &\\
	        & 0 & $\frac{1}{\sqrt{2}}(\ket{1}\!\bra{1} \!-\! \ket{2}\!\bra{2})$  not $\perp$  & \\
	        & 0-associated & $\ket{0}\!\bra{2}, \ket{2}\!\bra{0}$ & $\ket{1}\!\bra{2}, \ket{2}\!\bra{1}, \ket{1}\!\bra{1}$&  generalized eigenvectors\\
\hline	        
\end{tabular}\end{table*}

  Left-eigenvectors with eigenvalues of 1 are conserved quantities.   The identity is always a conserved quantity. 
For channel $F_{2a}$ $\ket{0}\!\bra{0} + \ket{1}\!\bra{1}$ is also a conserved quantity, 
  which implies that the trace of the subspace spanned by $\ket{0}, \ket{1}$ 
  is preserved via the channel.   
The $D_{2a}$ channel has   conserved quantity $\ket{0}\!\bra{1} +  \ket{1}\!\bra{0}$ in addition   to $\ket{0}\!\bra{0} +  \ket{1}\!\bra{1}$.  The two conserved quantities within    the subspace spanned by $\{ \ket{0}, \ket{1}\}$ are present because the channel gives a unitary transformation within this subspace.    

\addtocounter{table}{-1}
	        
\begin{table*}[ht]\footnotesize
\caption{Properties of function-generated channels for 3 state systems -continued}
\begin{tabular}{lllll}
\hline
Channel & eigenvalue & r-eigenvec. & l-eigenvec. & notes \\
\hline	
$A_3$ & 1  & $\ket{1}\!\bra{1} $ & $\ket{0}\!\bra{0} \!+\! \ket{1}\!\bra{1}\! +\!\ket{2}\!\bra{2} $ & Dephasing \\ 
            & 0 & $\frac{1}{\sqrt{2}}(\ket{0}\!\bra{0} \!-\! \ket{1}\!\bra{1})$ & $\ket{0}\!\bra{0}$  & Ergodic, Mixing \\
        & 0 &    $\frac{1}{\sqrt{2}}(\ket{2}\!\bra{2} \!-\! \ket{1}\!\bra{1})$ & $ \ket{2}\!\bra{2}$  &  \\
            & 0 & $\ket{i}\!\bra{j} $ for $ i\ne j$ &same  &   \\   
$B_3$  &  1 & $\ket{i}\bra{i}$ $\forall i$ & ${\cal E} = {\cal E}^\ddagger$&Dephasing\\
	    &  0 & $ \ket{i}\!\bra{j}$ for $ i \ne j$ & &  Unital \\
$C_3$  & $\omega^k$  $k \in {\mathbb Z}_3$ & $\frac{1}{\sqrt{3}}\sum_j \omega^{-jk} \ket{j}\!\bra{j}$ & ${\cal E} \sim  {\cal E}^\ddagger$  via permutation  & 3-cycle, $\omega=e^{2\pi i/3}$\\
                          & 0 & $\ket{i}\!\bra{j} $ for $ i\ne j$ & &  Dephasing, Unital \\
$D_3$ & $\pm 1$  & $\frac{1}{\sqrt{2}}(\ket{0}\!\bra{0} \!\pm\! \ket{1}\!\bra{1})$ & $\frac{1}{\sqrt{2}}(\ket{0}\!\bra{0} \!\pm \!\ket{1}\!\bra{1}\! + \!\ket{2}\!\bra{2})$& 2-cycle \\
            &  0            & $\frac{1}{\sqrt{2}}(\ket{2}\!\bra{2} \!-\! \ket{0}\!\bra{0}) $  not $\perp$&$\ket{2}\!\bra{2}$ & Dephasing \\
            &  0           & $\ket{i}\!\bra{j}$ for $ i\ne j$ & same & \\
$E_3$ & $\pm 1$  & $\frac{1}{\sqrt{2}}(\ket{0}\!\bra{0} \!\pm \!\ket{1}\!\bra{1})$ & ${\cal E} = {\cal E}^\ddagger$& 2-cycle \\
           & 1         & $\ket{2}\!\bra{2}$ & &  Dephasing, Unital\\
           & 0 & $\ket{i}\!\bra{j}$ for $i\ne j$ \\
$F_3$ &  1 & $ \ket{1}\!\bra{1}, \ket{2}\!\bra{2} $ & 
                  $\ket{0}\!\bra{0}\!+\!\ket{1}\!\bra{1},\ket{2}\!\bra{2}$ & Dephasing \\
              & 0 & $\frac{1}{\sqrt{2}}(\ket{0}\!\bra{0}\! -\! \ket{1}\!\bra{1})$  not $\perp$ & $\ket{0}\!\bra{0}$ \\\
             & 0 & $\ket{i}\!\bra{j}$ for $i\ne j$  & same & \\
$G_3$ &  1 & $\ket{2}\!\bra{2}$ & $\ket{0}\!\bra{0} \!+\! \ket{1}\!\bra{1}\!+\! \ket{2}\!\bra{2}$ & Dephasing    \\
               &  0 & $\frac{1}{\sqrt{2}}(\ket{1}\!\bra{1}\! -\! \ket{2}\!\bra{2}) $  not $\perp$ & $\ket{0}\!\bra{0}$ & Ergodic, Mixing  \\
              & 0-associated & $\frac{1}{\sqrt{2}}(\ket{0}\!\bra{0} \!-\! \ket{2}\!\bra{2}) $ & $\ket{1}\!\bra{1}$& generalized l-eigenvector \\
	& 0 & $\ket{i}\bra{j}$ for $i\ne j$  & same & \\
\hline
\end{tabular}
\end{table*}

\section{Eigenvalues of a discrete function-generated channel}
\label{ap:proofs}

\begin{theorem}
Consider 
 ${\cal L}_{\cal E}$ the matrix representation of a channel ${\cal E}_{S_D,f}$ generated from a function $f$ and collection of disjoint sets $S_D$
defined as in Definition \ref{def:chan}.  All eigenvalues
of  ${\cal L}_{\cal E}$ are either 0 or have modulus unity.  
\end{theorem}

\begin{proof}
We work in the orthonormal basis $\ket{i}\!\bra{k}$ (with $i,k \in {\mathbb Z}_N$) for operators with respect
to the Hilbert-Schmidt or Frobenius inner product.  
Here $N$ is the dimension of the Hilbert space. 
We consider how the channel affects diagonal basis operators first.

\subsection{Diagonal basis vectors}

For every $i \in  {\mathbb Z}_N$, $i$ is either a fixed point of the function ($f(i)=i$),
is a member of a cycle ($f^k(i)=i$ for $k>1$) or there exists a $k>1$ such that
$f^k(i)$ is either a fixed point or a member of a cycle. We consider each of these
cases separately.  In each case we find a right-eigenstate of ${\cal L}_{\cal E}$.

If $f(i) = i$ is a fixed point of the function, 
then $\ket{i}\!\bra{i}$ is a fixed point of the channel with eigenvalue 1.
This follows as the index $i$ belongs to one of the sets $S_j$ and application of $K_j$ gives ${\cal E}(\ket{i}\!\bra{i}) = K_j \ket{i}\!\bra{i} K_j^\dagger = \ket{i}\!\bra{i}$. 

If $i$ is a member of a $k$-cycle of $f$, with $k>1$ then there are $k$ associated
eigenvectors and eigenvalues of the channel. 
We find the set $S_j$ that contains $i$.
Application of the channel 
${\cal E}(\ket{i}\!\bra{i} )= K_j \ket{i}\!\bra{i} K_j^\dagger = \ket{f(i)}\!\bra{f(i)}$.
By applying the channel a second time, we obtain $\ket{f^2(i)}\!\bra{f^2(i)}$.
We apply the channel iteratively until we obtain the original state $\ket{i}\bra{i}$.  
A total of 
$k$ orthogonal (via the Hilbert-Schmidt inner product) right-eigenstates of ${\cal L}_{\cal E}$ can be constructed via discrete Fourier transform of the intermediate
states in the cycle and using powers of $\omega = e^{2\pi i/k}$ (see   equation \ref{eqn:rhom}). 
Eigenvalues are complex roots of unity which are integer powers of $\omega$.
The members of the cycle ($j$ such that $f^k(j)=j$ with $k>1$) are different from the fixed points of the function ($i$ such that $f(i)=i$), 
so eigenstates generated from a cycle are orthogonal to eigenstates generated by 
fixed points.   Elements in a cycle or that are fixed points are disjoint from those 
that are not in a cycle or are a fixed point.  Diagonal states 
that are formed via elements that are neither in a cycle or are a fixed 
point would generate  diagonal states that are perpendicular to those generated 
from fixed points or elements of cycles. 
For each $k$ cycle, the quantum Fourier transform gives $k$ linearly independent right-eigenvectors that are sums 
of diagonal operators. 

The remaining integers $i$ are those where 
 $i$ is not a fixed point of the function $f$ and $i$ is not a member
of a cycle.  

If $f(i)$ is a fixed point of $f$, (satisfying $f^2(i) = f(i)$) but $i$ is not a fixed point ($f(i) \ne i$), then 
$\ket{i}\!\bra{i} - \ket{f(i)}\!\bra{f(i)}$ is a right-eigenvector with eigenvalue 0. 
Note that it is not orthogonal (in the sense of the Hilbert-Schmidt inner product)
to the eigenstate associated with the fixed point $\ket{f(i)}\!\bra{f(i)}$.   
However this eigenstate is
linearly independent of the right-eigenstates of ${\cal L}_{\cal E}$
associated with fixed points and with cycles
of the function. 

If $f^2(i)$ is a fixed point of $f$ ($f^3(i) = f^2(i)$) but $f(i)$ is not a fixed point ($f^2(i) \ne f(i)$), then consider 
the operator  
$b = \ket{i}\!\bra{i} - \ket{f^2(i)}\!\bra{f^2(i)}$.
We operate on $b$ with the channel giving
\begin{align*}
{\cal E}(b) &= \ket{f(i)}\!\bra{f(i)} - \ket{f^2(i)}\!\bra{f^2(i)} \\
{\cal E}^2(b) &= 0. 
\end{align*}
The operator ${\cal E}(b)$ is a right-eigenvector of ${\cal L}_{\cal E}$ with eigenvalue 0.
This means $b$ is an generalized right-eigenvector associated with right-eigenvector
 ${\cal E}(b)$ that has a zero eigenvalue. 
This generalized eigenstate is
linearly independent of the previously identified right-eigenstates associated with fixed points and cycles of the function $f$. 

If $f(i)$ is a member of a $k$-cycle of $f$ with $k>1$, and $i$ is not a member
of the cycle, then  
$\ket{i}\!\bra{i} - \ket{f^{k-1}(i)}\!\bra{f^{k-1}(i) }$  is a right-eigenvector with eigenvalue 0. 
This eigenstate is
linearly independent of the right-eigenstates we discussed above that are associated with fixed points and cycles of the function.

If $f^2(i)$ is a member of a $k$-cycle of $f$ with $k>1$, and $f(i)$ is
not in the cycle,  then consider 
the state 
$b = \ket{i}\!\bra{i} - \ket{f(i)}\!\bra{f(i)}$.
We operate on $b$ with the channel giving
${\cal E}(b) = \ket{f(i)}\!\bra{f(i)} - \ket{f^2(i)}\!\bra{f^2(i)}$.
This is the form previously considered above and is
a right-eigenvector which has an eigenvalue of 0.  This implies that 
$b$ is a generalized eigenvector associated with right-eigenvector
 ${\cal E}(b) $ that has a zero eigenvalue. 
 
If $f^k(i)$ the smallest positive $k$ with $k>1$ such that $f^k(i)$ is either
 a member of a cycle or is a fixed point, then 
 $b = \ket{i}\!\bra{i} - \ket{f(i)}\!\bra{f(i)}$ is a generalized eigenvector.
 This follows iteratively using the construction used 
 above as ${\cal E}(b)$ is either a right-eigenvector
 with eigenvalue 0 or a generalized eigenvector with eigenvalue 0. 
 
Altogether we can find a unique eigenvector comprised of states 
on the diagonal for each $i$ value.   If we are careful not to over count members of cycles, 
we find that the set contains $N$ linearly independent 
eigenvectors and generalized eigenvectors comprised of sums of diagonal operators. 
The eigenvectors are independent of the disjoint subsets $\{ S_j\}$ used to generate 
the channel. 
The types of eigenvectors are illustrated in the graph shown in Figure \ref{fig:graph} 
of a function which makes it clearer 
why there are $N$ eigenvalues.   The construction of each type discussed above 
shows that they are linearly independent.  
 All the generalized eigenvectors
 are associated with eigenvalues of 0.  Only members of cycles or
 fixed points generate eigenvectors with eigenvalues that are roots of unity.   
 
 \begin{figure}[ht]\centering
 \includegraphics[width =3 truein]{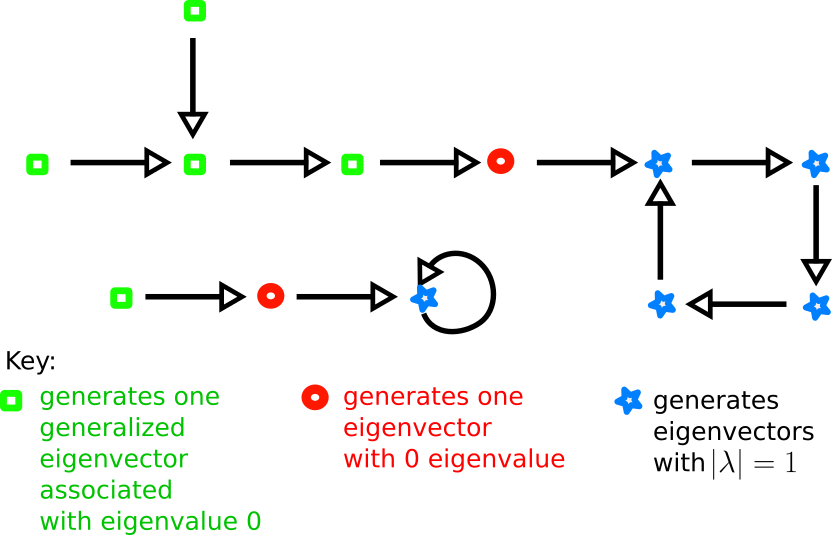}
 \caption{
 We show graphs of a function $f$ and how its structure is related 
 to right-eigenvectors that are comprised of diagonal operators.
 The right eigenvectors are those of ${\cal L}_{\cal E}$ the matrix representation of  
 of a quantum channel that is generated by a function and a set of disjoint sets.
 Each arrow begins at an element $x$ and points to another element $f(x)$. 
 Elements shown as red circles generate a right eigenvector with eigenvalue 0.
 Elements  shown as blue stars generate right eigenvectors with eigenvalue that
 have $|\lambda|=1$. Elements shown with green square generate generalized eigenvectors with eigenvalue 0.   The right eigenvectors comprised of diagonal 
 elements are insensitive to the disjoint sets $\{S_j\}$.  
  \label{fig:graph} }
 \end{figure}
 
 \begin{figure}[ht]\centering
 \includegraphics[width =3.3 truein]{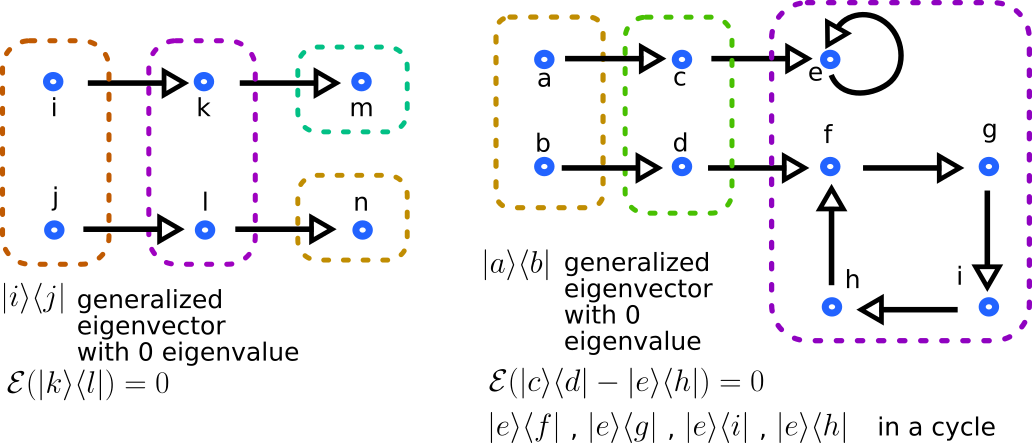}
 \caption{
 We show graphs of a function $f$ and how its structure is related 
 to right-eigenvectors that are comprised of off-diagonal operators.
 The dashed lines show different disjoint sets. 
 On the left $\ket{k}\!\bra{l}$ (with $k \ne l$) and its transpose have eigenvalue 0 and $\ket{i}\!\bra{j}$  ($i \ne j$) 
 and its transpose are generalized eigenvectors with eigenvalue 0. 
 On the right $\ket{a}\!\bra{b}$ (with $a\ne b$) and its transpose are  
 generalized eigenvectors with eigenvalue 0
 and $\ket{c}\!\bra{d}$  ($c\ne d$) and its transpose are right eigenvectors with eigenvalue 0.
 The operators $\ket{e}\!\bra{f}$, $\ket{e}\!\bra{g}$, $\ket{e}\!\bra{h}$,
 and $\ket{e}\!\bra{h}$ (with $e,f,g,h$ all distinct) and are in a cycle and 
 so generate 4 right eigenvectors with eigenvalues that are complex roots of unity.    \label{fig:graphij} }
 \end{figure}

\subsection{Off-diagonal basis vectors} 

We now consider how the channel affects operators comprised of off-diagonal 
terms. 
We consider how the channel operates on $\ket{i}\!\bra{k}$ where $i\ne k$. 
We consider two cases, $i,k$ are both contained in one disjoint set $S_j$
and where there is no disjoint set in $S_D$ containing both $i,k$.

Suppose
there is no disjoint set $S_j \in S_D$ that contains both $i,k$.
We find that $K_j \ket{i}\!\bra{k} K_j^\dagger = 0$
for all $j$ and this implies that  ${\cal E}(\ket{i}\bra{k}) = 0$.
The operator $\ket{i}\!\bra{k}$ is a right-eigenvector 
of ${\cal L}_{\cal E}$ with eigenvalue 0.

Consider the second case for $i,k$ where $i \ne k$ and both are contained
in a disjoint set $S_j$. 
\begin{align*}
{\cal E}( \ket{i}\!\bra{k}  ) \!
&=\! K_j  \ket{i}\!\bra{k} K_j^\dagger\\
&=\! (\ket{f(i)}\!\bra{i}\!+\! \ket{f(k)}\!\bra{k} ) \ket{i}\!\bra{k}  (\ket{i}\!\bra{f(i)}\! +\! \ket{k}\!\bra{f(k)})\\
& = \ket{f(i)}\! \bra{f(k)} .
\end{align*}
The matrices $\ket{i}\!\bra{k}, \ket{k}\!\bra{i}$ are generalized right-eigenvectors associated with eigenvalue 0 unless both $f(i), f(j)$ are in the same disjoint set 
$f(i), f(k) \in S_m$ for some $m$.
Here $S_m$ could be the same disjoint set as $S_j$. 
Note that $f(i) \ne f(k)$ due to a requirement on the nature of the disjoint sets in 
the definition of the channel (Definition \ref{def:chan}).  
Thus ${\cal E}(\ket{i}\!\bra{k})$ must give an off-diagonal matrix. 
Suppose $f(i)=i$   and $f(k) = k$ and $i,j \in S_j$.  
This case $\ket{i}\!\bra{k}$ and  $\ket{k}\!\bra{i}$ are both 
right eigenvectors with eigenvalue of 1. 
If both $f(i), f(k) \in S_m$ (and $i,k$ are not fixed points and  $i,j \in S_j$) 
we can apply the channel iteratively. 
The next iteration gives 
${\cal E}^2(\ket{i}\!\bra{k}) =  \ket{f^2(i)}\!\bra{f^2(k)} .$
Again, if both $f^2(i), f^2(k) $ are contained in different
disjoint sets  then $\ket{i}\!\bra{k}$ is a generalized
eigenvector associated with eigenvalue 0. 

If $f^m(i)=i$, $f^m(k)=k$ for a positive integer $m>1$ and 
for all positive integers $n$ 
there exists a disjoint 
set $S_j$ such that both 
$f^n(i), f^n(k) \in S_j$,   then 
$\ket{i}\!\bra{k}$ generates a cycle. 
From $\ket{i}\!\bra{k}$ and iterates of it,  we can construct a set of 
right-eigenvectors (comprised only of sums of off-diagonal operators) with eigenvalues of modulus 1 using a Fourier transform like that given in equation \ref{eqn:rhom}.
The resulting right-eigenvectors 
are linearly independent and would be linearly independent of off diagonal 
operators that have 0 eigenvalues or that are generalized eigenvectors
associated with a 0 eigenvalue.  
The only other way 
 iteration of $\ket{i}\!\bra{k}$ can terminate (other than terminating via a cycle) is when an iteration of $i$ and $k$ by 
 the function 
gives a pair of elements that are not contained in the same disjoint set.  In this case we have a series of generalized right-eigenvectors associated with eigenvalue 0.   
Examples of the construction of eigenvectors for off-diagonal operators are shown 
 in figure \ref{fig:graphij}.

Each off diagonal term, is either a right eigenvector 
with eigenvalue 0, a generalized right eigenvector associated with eigenvalue 0,
a fixed point
or generates a cycle giving eigenvectors with eigenvalues that are complex roots of unity.  In total we can generate $N(N-1)$ linearly independent 
eigenvectors from each pair off-diagonal elements.  
The union of the $N$ eigenvectors comprised of diagonal operators 
and the $N(N-1)$ eigenvectors comprised of off-diagonal operators 
is a set of $N^2$ linearly independent right eigenvectors. 

We have shown iteratively that we can find a full set ($N^2$) of eigenvalues
of ${\cal L}_{\cal E}$ where
eigenvalues either are roots of unity, arising from the fixed points and cycles of $f$,  
or they have eigenvalue 0.  Zero eigenvalues are associated with 
both right-eigenvectors and generalized right-eigenvectors. In both cases they 
arise from orbits of $f$ that, after iteration, have end-states
that are fixed points or cycles. 
Right-eigenvectors or generalized right-eigenvectors 
are sometimes are not orthogonal (in the sense
of the Hilbert-Schmidt inner product) to each other. However,
we have identified a set of $N^2$ right-eigenvectors and generalized
right-eigenvectors that are linearly independent. 
The Jordan form of ${\cal L}_{\cal E}$ consists of Jordan blocks of dimension 1
that contain eigenvalues of modulus 1, and Jordan blocks that can have a larger dimension that have eigenvalues of 0.   As the non-trivial Jordan blocks
are associated with chains of $f$ that have end-states that are  
fixed points or cycles,  the maximum dimension of any Jordan block is the maximum
link number $k_c$ (as defined in equation \ref{eqn:k_c}).
As every Jordan block of ${\cal L}_{\cal E}$ generates an invariant subspace of the 
${\cal L}_{\cal E}$, and non-zero
eigenvalues have modulus of 1, the asymptotic subspace 
of the channel is spanned by the set of right eigenvectors with 
non-zero eigenvalues. 
\end{proof}

\end{document}